\documentclass[acmsmall]{acmart}

\AtBeginDocument{%
  \providecommand\BibTeX{{%
    \normalfont B\kern-0.5em{\scshape i\kern-0.25em b}\kern-0.8em\TeX}}}

\copyrightyear{2021}
\acmYear{2021}
\setcopyright{acmcopyright}\acmConference[EC '21]{Proceedings of the 22nd ACM Conference on Economics and Computation}{July 18--23, 2021}{Budapest, Hungary}
\acmBooktitle{Proceedings of the 22nd ACM Conference on Economics and Computation (EC '21), July 18--23, 2021, Budapest, Hungary}
\acmPrice{15.00}
\acmDOI{10.1145/3465456.3467634}
\acmISBN{978-1-4503-8554-1/21/07}

\usepackage{booktabs} 
\usepackage{url}

\usepackage{amsthm,amsmath}
\usepackage{enumerate}
\usepackage{xspace}
\usepackage{subcaption}
\usepackage[labelformat=simple]{subcaption}

\usepackage[switch]{lineno}
\usepackage{algorithm, float, setspace}
\usepackage[noend]{algpseudocode}
\usepackage{multirow}
\usepackage{hhline}

\setcitestyle{acmnumeric}
\usepackage{mathtools}
\usepackage{comment}

\usepackage{xcolor}

\theoremstyle{plain}
\newtheorem{theorem}{Theorem}
\newtheorem{lemma}{Lemma}
\newtheorem{corollary}{Corollary}[theorem]
\newtheorem{example}{Example}
\newtheorem{definition}{Definition}
\newtheorem{claim}{Claim}
\DeclareRobustCommand\iff{\;\Longleftrightarrow\;}

\newcommand{\Mech}[1]{Mechanism~\ref{mech:#1}}
\newcommand{\Algo}[1]{Algorithm~\ref{algo:#1}}

\newcommand{\Ex}[1]{Example~\ref{ex:#1}}
\newcommand{\Thm}[1]{Theorem~\ref{thm:#1}}

\newcommand{\snp}{S\&P\,500\xspace}
\newcommand{\AAPL}{\texttt{AAPL}\xspace}
\newcommand{\MSFT}{\texttt{MSFT}\xspace}
\newcommand{\DIS}{\texttt{DIS}\xspace}
\newcommand{\valpham}{\boldsymbol{\alpha}_m}
\newcommand{\valpha}{\boldsymbol{\alpha}}
\newcommand{\vbetan}{\boldsymbol{\beta}_n}
\newcommand{\vbeta}{\boldsymbol{\beta}}
\newcommand{\vomega}{\boldsymbol{\omega}}
\renewcommand{\S}{\boldsymbol{S}}
\newcommand{\vphi}{\boldsymbol{\phi}}
\newcommand{\vpsi}{\boldsymbol{\psi}}
\newcommand{\vp}{\boldsymbol{p}}
\newcommand{\vb}{\boldsymbol{b}}
\newcommand{\vq}{{\boldsymbol{q}}}
\newcommand{\va}{\boldsymbol{a}}
\newcommand{\vgamma}{\boldsymbol{\gamma}}
\newcommand{\vdelta}{\boldsymbol{\delta}}
\newcommand{\order}{\mathcal{O}}
\newcommand{\vf}{\boldsymbol{f}}
\newcommand{\vg}{\boldsymbol{g}}
\newcommand{\vI}{\boldsymbol{I}}
\newcommand{\I}{\mathcal{I}}
\newcommand{\M}{\mathcal{M}}

\newcommand{\K}{\mathcal{K}}
\newcommand{\vone}{\mathbf{1}}
\newcommand{\Cc}{\mathcal{C}}
\newcommand{\F}{\mathcal{F}}
\newcommand{\cP}{\mathcal{P}}

\newcommand{\R}{\mathbb{R}}

\newcommand{\vzero}{\boldsymbol{0}}

\newcommand{\card}[1]{\lvert#1\rvert}

\title[Designing a Combinatorial Financial Options Market]{Designing a Combinatorial Financial Options Market}

\author{Xintong Wang}
\affiliation{Harvard University}
\email{xintongw@seas.harvard.edu}
\orcid{0000-0002-0867-8807}

\author{David M. Pennock}
\affiliation{Rutgers University}
\email{dpennock@dimacs.rutgers.edu}

\author{Nikhil R. Devanur}
\affiliation{Amazon}
\email{ndevanur@gmail.com}

\author{David M. Rothschild}
\affiliation{Microsoft Research}
\email{davidmr@microsoft.com}

\author{Biaoshuai Tao}
\affiliation{Shanghai Jiao Tong University}
\email{bstao@sjtu.edu.cn}

\author{Michael P. Wellman}
\affiliation{University of Michigan}
\email{wellman@umich.edu}
\orcid{0000-0002-1691-6844}

\begin{abstract}
\textit{Financial options} are contracts that specify the right to buy or sell an underlying asset at a \textit{strike price} by an expiration date.
Standard exchanges offer options of predetermined strike values and trade options of different strikes independently, even for those written on the same underlying asset.
Such independent market design can introduce arbitrage opportunities and lead to the thin market problem.
The paper first proposes a mechanism that consolidates and matches orders on standard options related to the same underlying asset, while providing agents the flexibility to specify any custom strike value.
The mechanism generalizes the classic double auction, runs in time polynomial to the number of orders, and poses no risk to the exchange, regardless of the value of the underlying asset at expiration.
Empirical analysis on real-market options data shows that the mechanism can find new matches for options of different strike prices and reduce bid-ask spreads.\looseness=-1

Extending standard options written on a \textit{single} asset, we propose and define a new derivative instrument --- \textit{combinatorial financial options} that offer contract holders the right to buy or sell any linear combination of multiple underlying assets.
We generalize our single-asset mechanism to match options written on different combinations of assets, and prove that optimal clearing of combinatorial financial options is coNP-hard. 
To facilitate market operations, we propose an algorithm that finds the exact optimal match through iterative constraint generation, and evaluate its performance on synthetically generated combinatorial options markets of different scales.
As option prices reveal the market's collective belief of an underlying asset's future value, a combinatorial options market enables the expression of aggregate belief about future correlations among assets.\looseness=-1
\end{abstract}



\begin{document}


\maketitle


\section{Introduction}
Financial options are securities that 
convey rights to conduct specific trades in the future.
For example, an \snp \emph{call option} with \emph{strike price} 4500 and \emph{expiration date} December 31, 2021, provides the right to buy one share of \snp at $\$4500$ on that day.%
\footnote{For simplicity, examples are given without the typical 100x contract multiplier.}
A standard financial option is a derivative instrument of a \emph{single} underlying asset or index, and its payoff is a function of the underlying variable. 
The example option contract pays $\max\{S-4500,0\}$, where $S$ is the value of the \snp index on the expiration date.%
\footnote{
	Throughout this paper, we restrict to \emph{European options} which can be exercised only at expiration.
	The settlement value is calculated as the opening value of the index on the expiration date or the last business day (usually a Friday) before the expiration date.
	In cash-settled markets, instead of actual physical delivery of the underlying asset, the option holder gets a cash payment that is equivalent to the value of the asset.
} 
Investors trade options to hedge risks and achieve certain return patterns, or to speculate about the movement of the underlying asset price, buying an option when its price falls below their estimate of its expected value. 
Thus, option prices reveal the collective risk-neutral belief distribution of the underlying asset's future value.

Despite the significant volume of trade and interest in option contracts, financial exchanges that support trading of standard options have two limitations that compromise its expressiveness and efficiency.
\textit{First}, the exchange offers only a selective set of markets for options written on a specific underlying asset (i.e., option chain), in which each market features a predetermined strike price and expiration date.
For example, as of this writing, the Chicago Board Options Exchange (CBOE) offers 120 distinct strike prices, ranging from 1000 to 6000 at intervals of 25, 50, or 100, for \snp options expiring by the end of 2021. 
While traders can engineer custom contracts (e.g., a \snp call option with a strike price of \$4520) by simultaneously purchasing multiple available options in appropriate proportions, they must monitor several markets to ensure that a bundle can be constructed at a desired price. 
Often, execution risk and transaction costs prevent individual traders from carrying out such strategies. 
As the prescriptive markets prevent traders from conveniently expressing custom strike values, the exchange may fail to aggregate supply and demand requests of greater detail, leading to a loss of economic efficiency.

\textit{Second}, each market within an option chain independently aggregates and matches orders in regard to a single contract with a specified option type, strike price, and expiration date, despite its interconnectedness to other option markets and their common dependency on the underlying asset.
Such independent market design fails to notify traders of strictly better, cheaper options 
and generically allows arbitrage opportunities.
Moreover, investments get diluted across independent markets even when participants are interested in the same underlying asset.
This can lead to the problem of thin markets, where few trades happen and bid-ask spreads become wide.
Even for some of the most actively traded option families, empirical evidence has shown that liquidity can vary much across option types and strikes.
\citeauthor{Cao2010} studied eight years of options trading data and find consistently lower liquidity in puts and deep in-the-money options~\cite{Cao2010}.

\subsection{An Exchange for Standard Financial Options}
The paper first focuses on the design of an exchange for standard financial options to address limitations discussed above.
We propose a mechanism to aggregate and match orders on options across different strike prices that are logically related to the same underlying asset.
As a result, the exchange operates a consolidated market for each underlying security and expiration, where traders can specify custom options of arbitrary strike value.
The mechanism is an exchange that generalizes the double auction; it is not a market maker or any individual arbitrageur. 
It works by finding a match that maximizes net profit subject to zero loss even in the worst case, and thus poses no risk to the exchange regardless of the value of the underlying security at expiration.
We show that the mechanism is computationally efficient, and key market operations (i.e., match and price quotes) can be computed in time polynomial in the number of orders.
We conduct experiments on real-market options data, using our mechanism to consolidate outstanding orders from independently-traded options markets.
Empirical results demonstrate that the proposed mechanism can match options that the current exchange cannot and provide more competitive bid and ask prices.
The improved efficiency and expressiveness may help to aggregate more fine-grained information and recover a complete and fully general probability distribution of the underlying asset's future value at expiration.

\subsection{An Exchange for Combinatorial Financial Options}
In the second part of the work, we generalize standard financial options to define \textit{combinatorial financial options}---derivatives that give contract holders the right to buy or sell \emph{any linear combination} of underlying assets at any strike price.
For example, a call option written on ``$1$\AAPL$+1$\MSFT'' with strike price 300 specifies the right to buy one share of Apple \emph{and} one share of Microsoft at $\$300$ total price on the expiration date, whereas a call on ``$2$\AAPL$-1$\MSFT'' with strike price 50 confers the ability to buy two shares of Apple and sell one share of Microsoft at expiration for a net cost of \$50.
Combinatorial options allow traders to conveniently and precisely hedge their exact portfolio, replicating any payoff functions that standard options can achieve and exponentially more. 
Some of the standard options on mutual funds and stock indices, like the \snp and Dow Jones Industrial Average (DJI), are indeed combinatorial options written on the respectively predefined portfolios. 
The CBOE recently launched options on eleven Select Sector Indices,%
\footnote{\url{https://markets.cboe.com/tradable_products/sp_500/cboe_select_sectors_index_options/}}
each of which can be considered as a combinatorial option of a pre-specified linear combination of stocks representing one economic sector.
The goal of this work is to allow traders to create options on their own custom indices, representing any linear combination of stocks and sectors that they want.

We design an exchange for combinatorial financial options.
Such a market enables the elicitation and recovery of future correlations (e.g., complimentary or substitution effects) among assets. 
However, as traders are offered the expressiveness to specify weights (or shares) for multiple underlying assets, new challenges arise and the thin market problem exacerbates.
Opening a new, separate exchange for each new combination of stocks and weights would rapidly grow intractable.
Naively matching buy and sell orders for only the exact same portfolio may yield few or no trades, despite plenty of acceptable trades among orders. 

Extending the mechanism that consolidates standard options on a single underlying security, we propose an optimization formulation that can match combinatorial options written on different linear combinations of underlying assets.
The mechanism maximizes net profit subject to no risk to the exchange, regardless of (any combination of) values of all assets at expiration. 
We show that the proposed matching mechanism with increased expressiveness, however, comes at the cost of higher computational complexity: the optimal clearing of a combinatorial options market is coNP-hard. 

We demonstrate that the proposed mechanism can be solved relatively efficiently by exploiting constraint generation techniques.
We propose an algorithm that finds the exact optimal-matching solution by satisfying an increasing set of constraints iteratively generated from different future values of the underlying assets.
In experiments on synthetic combinatorial orders generated from real-market standard options prices, we show that the matching algorithm terminates quickly, with its running time growing linearly with the number of orders and the size of underlying assets. 

\subsection{Roadmap}
Section~\ref{sec:related} discusses related work on option pricing and combinatorial markets.
Section~\ref{sec:notations} introduces notations and background on the current options market.
We first focus on the market design for standard financial options. 
Section~\ref{sec:standard_option} introduces the formal setting of a consolidated options market with limit orders, and proposes a matching mechanism for options across types and strikes (\Mech{single_match}).
Section~\ref{sec:combo_options}, our main contribution, defines a combinatorial financial options exchange (Mechanism \ref{mech:combo_match}), analyzes the computational complexity of optimally clearing such a market (Theorem~\ref{thm:NP-complete} and Theorem~\ref{thm:coNP-hard}), and proposes a constraint generation algorithm to find the exact optimal match (Algorithm~\ref{algo:comb_match}). 
Section~\ref{sec:exp_standard_option} and Section~\ref{sec:exp_combo_option} evaluate the two proposed mechanisms on real-market standard options data and synthetically generated combinatorial options data, respectively.
Section~\ref{sec:conclusion} concludes and discusses possible future directions.

\section{Related Work}
\label{sec:related}
\subsection{Rational Option Pricing}
Our proposed market designs relate to arbitrage conditions that have been studied extensively in financial economics ~\cite{Modigliani1958,Varian1987}.
In short, an arbitrage describes the scenario of ``free lunches''---configurations of prices such that one can get something for nothing.
The matching operation of an exchange (i.e., the auctioneer function) can be considered as arbitrage elimination: matching orders in effect works by identifying combinations of orders that reflect a risk-free surplus (i.e., gains from trade). 
In the classic double auction, a match accepting the highest buy order at, say, \$12 and the lowest sell order at \$10 in effect yields an arbitrage surplus of \$2 that can go to the buyer, the seller, the exchange, or split among them in any way. 

\citeauthor{Merton1973} first investigated the no-arbitrage pricing for options, stating the necessity of convexity in option prices~\cite{Merton1973}.
Other relevant works examine no-arbitrage conditions for options under different scenarios, such as modeling the stochastic behavior of the underlying asset (e.g., the Black-Scholes model)~\cite{bs_model,pricing_discrete_time} and considering the presence of other types of securities (e.g., bonds and futures)~\cite{Garman1976, Ritchken1985}.
The most relevant work to our proposed approach is by \citeauthor{Herzel2005}, who makes little assumption on the underlying process of the asset price and other financial instruments~\cite{Herzel2005}.
The paper proposes a linear program to check the convexity between every strike-price pair; it finds arbitrage opportunities that yield positive cash flows now and no liabilities in the future on European options written on the same underlying security.

Our first contribution on a market design that consolidates standard options generalizes Herzel's by adding an extra degree of freedom to incorporate potential future payoff gain or loss into the matching. 
To our knowledge, no prior work has defined general combinatorial financial options or investigated the matching mechanism design and complexity for such a market.
\subsection{Combinatorial Market Design}
Much prior work examines the design of combinatorial markets, both exchanges and market makers \cite{Chen2008b,Hanson03}, for different applications including prediction markets with Boolean combinations \cite{Fortnow2005}, permutations \cite{Chen2007}, hierarchical structures \cite{Guo2009}, tournaments \cite{Chen2008a}, and electronic sourcing \cite{Sandholm2007}. 
\citeauthor{DudikEtAl13} show how to employ constraint generation in linear programming to keep complex related prices consistent~\cite{DudikEtAl13}. 
\citeauthor{KroerDudik16} generalize this approach using integer programming~\cite{KroerDudik16}.
\citeauthor{Rothschild14} investigate misaligned prices for logically related binary-payoff contracts in prediction market, and uncover persistent arbitrage opportunities for risk-neutral investors across different exchanges~\cite{Rothschild14}.

Designing combinatorial markets faces the tradeoff between expressiveness and computational complexity: 
giving participants 
greater flexibility to express preferences can help to elicit better information and increase economic efficiency, but in the meantime leads to a more intricate mechanism that is computationally harder. 
Several works have formally described and quantified such tradeoff \cite{Benisch2008,Golovin2007}, and studied how to balance it by exploiting the outcome space structure and limiting expressivity \cite{Chen2008b,XiaPe11,LaskeyEtAl18,Dudik2020}.
Our work here contributes to this rich literature by designing a vastly more expressive version of the popular financial options market. 
The state space and payoff function for combinatorial financial options differ from the previously studied combinatorial structures in both meaning and computational complexity. 
The associated matching problem is novel (with its new form of payoff function), and is not a specialization of previously studied combinatorial prediction market.

\section{Background and Notations}
\label{sec:notations}
There are two types of options, referred to as \emph{call} and \emph{put} options.
We denote a call option as $C(S, K, T)$ and a put option as $P(S, K, T)$, which respectively gives the option buyer the right to \emph{buy} and \emph{sell} an underlying asset $S$ at a specified \emph{strike price} $K$ on the \emph{expiration date} $T$.
In the rest of this paper, we omit $T$ from the tuples for simplicity, as the mechanism aggregates options within the same expiration.

The option buyer decides whether to exercise an option. 
Suppose that a buyer spends \$8 and purchases a call option, $C(\snp, 4500, 20211231)$.
If the \snp index is \$4700 at expiration, the buyer will pay the agreed strike \$4500, receive the index, and get a \textit{payoff} of \$200 and a \textit{net profit} of \$192 (assuming no time value). 
If the \snp price is \$4200, the buyer will walk away without exercising the option.
Therefore, the payoff of a purchased option is 
\[\Psi := \max\{\chi(S-K), 0\},\]
where $S$ is the value of underlying asset at expiration and $\chi \in \{-1, 1\}$ equals 1 for calls and $-1$ for puts.
As the payoff for a buyer is always non-negative, the seller receives a \emph{premium} now (e.g., \$8) to compensate for future obligations.

Option contracts written on the same underlying asset, type, strike, and expiration are referred to as an \emph{option series}.
Consider options of a single security offering both calls and puts, ten expiration dates and fifty strike prices.
All option series render a total of a thousand markets, with each maintaining a separate \emph{limit order book}.
In such a market, deciding the existence of a transaction takes $\order(1)$ time by comparing the best bid and ask prices, and matching an incoming order can take up to $\order(n_\text{o})$ time depending on its quantity, where $n_\text{o}$ is the number of orders on the opposite side of the order book.

\section{Consolidating Standard Financial Options}
\label{sec:standard_option}
This section proposes a mechanism to consolidate buy and sell orders on standard financial options written on the same underlying asset across different types and strike prices.
The model does not make any assumption on the option's pricing model or the stochastic behavior of the underlying security.  
%
\subsection{Matching Orders on Standard Financial Options}
Consider an option market in regard to a single underlying asset $S$ with an expiration date $T$\@.
Traders can specify orders on such options with any positive strike value.
The market has a set of buy orders, indexed by $m \in \{1, 2, ..., M\}$, and a set of sell orders, indexed by $n \in \{1, 2, ..., N\}$. 
Buy orders are represented by a type vector $\vphi \in \{-1, 1\}^M$ with each entry specifying a put or a call option, a strike price vector $\vp \in \R_{+}^M$, and bid prices $\vb \in \R_{+}^M$.
Sell orders are denoted by a separate type vector $\vpsi \in \{-1, 1\}^N$, a strike vector $\vq \in \R_{+}^N$, and ask prices $\va \in \R_{+}^N$.

The exchange aims to match buy and sell orders submitted by traders.
Specifically, it decides the fraction $\vgamma \in [0, 1]^M$ to sell to buy orders and the fraction $\vdelta \in [0, 1]^{N}$ to buy from sell orders. 
We start with a formulation that matches orders by finding a common form of arbitrage in which the exchange maximizes net profit at the time of order transaction, subject to \textit{zero} payoff loss in the future for all possible states of $\S$:
\begin{align}
	\max \limits_{\vgamma, \vdelta} & \displaystyle \quad \vb^\top \vgamma - \va^\top \vdelta \notag\\
	\text{s.t.} & \displaystyle \quad \,\,\underbrace{\!\!\sum_{m} \gamma_m \max\{\phi_m(S - p_m), 0\}\!\!}_{\Psi_{\Gamma}}\,\, - \,\,\underbrace{\!\!\sum_{n} \delta_n \max\{\psi_n(S - q_n), 0\}\!\!}_{\Psi_{\Delta}}\,\,\leq 0, \quad \forall S \in [0, \infty) \notag
\end{align}
Here, the first term $\Psi_{\Gamma}:=\sum_{m} \gamma_m \max\{\phi_m(S - p_m), 0\}$ in the constraint calculates the total payoff of sold options as a function of $S$, which is the obligation or liability of the exchange at the time of option expiration.
The second term $\Psi_{\Delta} := \sum_{n} \delta_n \max\{\psi_n(S - q_n), 0\}$ computes the total payoff of bought options, by which the exchange has the right to exercise.
The constraint guarantees that the liability of the exchange does not exceed its payoff, regardless of the value of $S$ on the expiration date.
We denote options bought by the exchange as Portfolio $\Delta$ and options sold as Portfolio $\Gamma$, and the constraint enforces that Portfolio $\Delta$ \emph{(weakly) dominates} Portfolio $\Gamma$.

We further extend the above formulation.
Instead of restricting the exchange to zero loss at expiration, we allow it to take into account the potential (worst-case) deficit or gain in the future.
We denote this extra degree of freedom as a decision variable $L \in \R$, and have the following matching mechanism:
%
\begin{align}\tag{M.1}
\label{mech:single_match}
\max \limits_{\vgamma, \vdelta, L} & \displaystyle \quad \vb^\top \vgamma - \va^\top \vdelta - L\\
\text{s.t.} & \displaystyle \quad \,\,\underbrace{\!\!\sum_{m} \gamma_m \max\{\phi_m(S - p_m), 0\}\!\!}_{\Psi_{\Gamma}}\,\, - \,\,\underbrace{\!\!\sum_{n} \delta_n \max\{\psi_n(S - q_n), 0\}\!\!}_{\Psi_{\Delta}}\,\,\leq L, \quad \forall S \in [0, \infty) \label{eq:single_constraint}
\end{align}

The constraint now guarantees that the difference between the liability and the payoff of the exchange will be bounded by $L$, regardless of the value of $S$ on the expiration date. 
Following the definition below, we say that the constraint enforces Portfolio $\Delta$ to \emph{(weakly) dominate} Portfolio $\Gamma$ with some constant offset $L$.
%
\begin{definition} [Payoff dominance with an offset. Adapted from~\citeauthor{Merton1973}~\cite{Merton1973}]
	\label{def:opt_dominate}
	Portfolio $\Delta$ \emph{(weakly) dominates} Portfolio $\Gamma$ with an offset $L$, if the payoff of Portfolio $\Delta$ plus a constant $L$ is 
	greater than or equal to that of Portfolio $\Gamma$ for all possible states of the underlying variable at expiration. 
	Portfolio $\Gamma$ is said to be \emph{(weakly) dominated} by Portfolio $\Delta$ with an offset $L$.
\end{definition}
%
Since this potential gain or loss is further incorporated in the objective at the time of matching, the mechanism \ref{mech:single_match}  guarantees no overall loss for each match, and enjoys an extra degree of freedom to trade.
We give two motivating examples based on real-market options data to illustrate the economic meaning and the usefulness of a flexible $L$.
Each of the matches shown below would not have been found if we restrict a zero payoff loss for the exchange.
\Ex{nonneg_L} showcases the scenario where $L$ allows the exchange to take a (worst-case) deficit at the time of option expiration, if it is preemptively covered by a surplus (i.e., revenue) at the time of order transaction.
\begin{example} [A match with a positive $L$] 
	\label{ex:nonneg_L}
	We use \ref{mech:single_match} to consolidate options of Walt Disney Co. (\DIS) that are priced on January 23, 2019 and expire on June 21, 2019.
	We find the following match, where each order would not transact in its corresponding independent market. 
	The exchange can
	\begin{itemize}
		\setlength{\itemsep}{2pt}
		\item sell to the buy order on $C(\DIS, 110)$ at bid \$7.2,
		\item sell to the buy order on $P(\DIS, 150)$ at bid \$38.75,
		\item buy from the sell order on $C(\DIS, 150)$ at ask \$0.05,
		\item buy from the sell order on $P(\DIS, 110)$ at ask \$5.1.
	\end{itemize}
	Therefore, the exchange gets an immediate gain of \$40.8 ($7.2+38.75-5.1-0.05$).
	Figure~\ref{subfig:lgtzero} plots the payoffs of bought and sold options as a function of DIS, showing that the exchange will have a net liability of $\$40$ (i.e., $L=40$), regardless of the DIS value at expiration.
	The exchange makes a net profit of \$0.80 from the match at no risk.
	\qed
\end{example}
\Ex{neg_L} below demonstrates a different scenario where the decision variable $L$ allows the exchange to have a temporary deficit (i.e., expense) at the time of order transaction, if it is guaranteed to earn it back later at the time of option expiration. 
\begin{example} [A match with a negative $L$]
	\label{ex:neg_L}
	We consolidate options of Apple Inc. (\AAPL) that are priced on January 23, 2019 and expire on January 17, 2020.
	We find the following match, where the exchange can
	\begin{itemize}
		\setlength{\itemsep}{2pt}
		\item sell to the buy order on $C(\AAPL, 160)$ at bid \$14.1,
		\item sell to the buy order on $P(\AAPL, 80)$ at bid \$0.62,
		\item buy from the sell order on $C(\AAPL, 80)$ at ask \$74.2,
		\item buy from the sell order on $P(\AAPL, 160)$ at ask \$19.1.
	\end{itemize}
	The match incurs an expense of \$78.58 now ($14.1+0.62-74.2-19.1$), and yields a guaranteed payoff of \$80 (i.e., $L=-80$) at expiration. 
	Figure~\ref{subfig:lstzero} depicts the respective payoffs of bought and sold options.
	From the match, we can infer an interest rate of 1.82\%, calculated by $78.58 e^{r \Delta t} = 80$.
	\qed
\end{example}
\begin{figure}[t]
	\centering
	\begin{subfigure}{0.47\columnwidth}	
		\includegraphics[width=0.9\columnwidth]{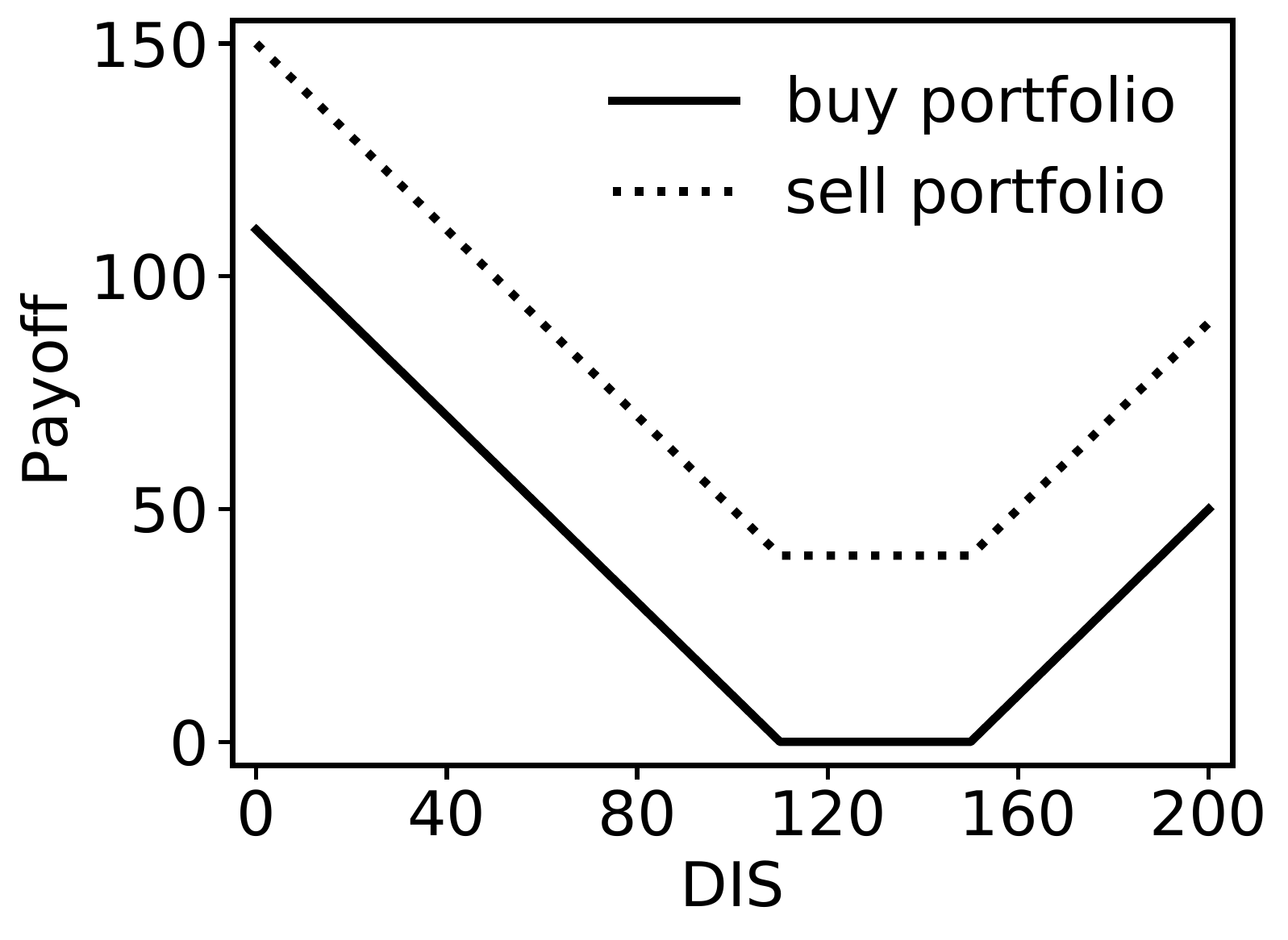}
		\caption{Payoffs of standard options bought and sold in \Ex{nonneg_L} as a function of DIS value.}
		\label{subfig:lgtzero}
	\end{subfigure}
	\hspace{0.04\columnwidth}
	\begin{subfigure}{0.47\columnwidth}	
		\includegraphics[width=0.9\columnwidth]{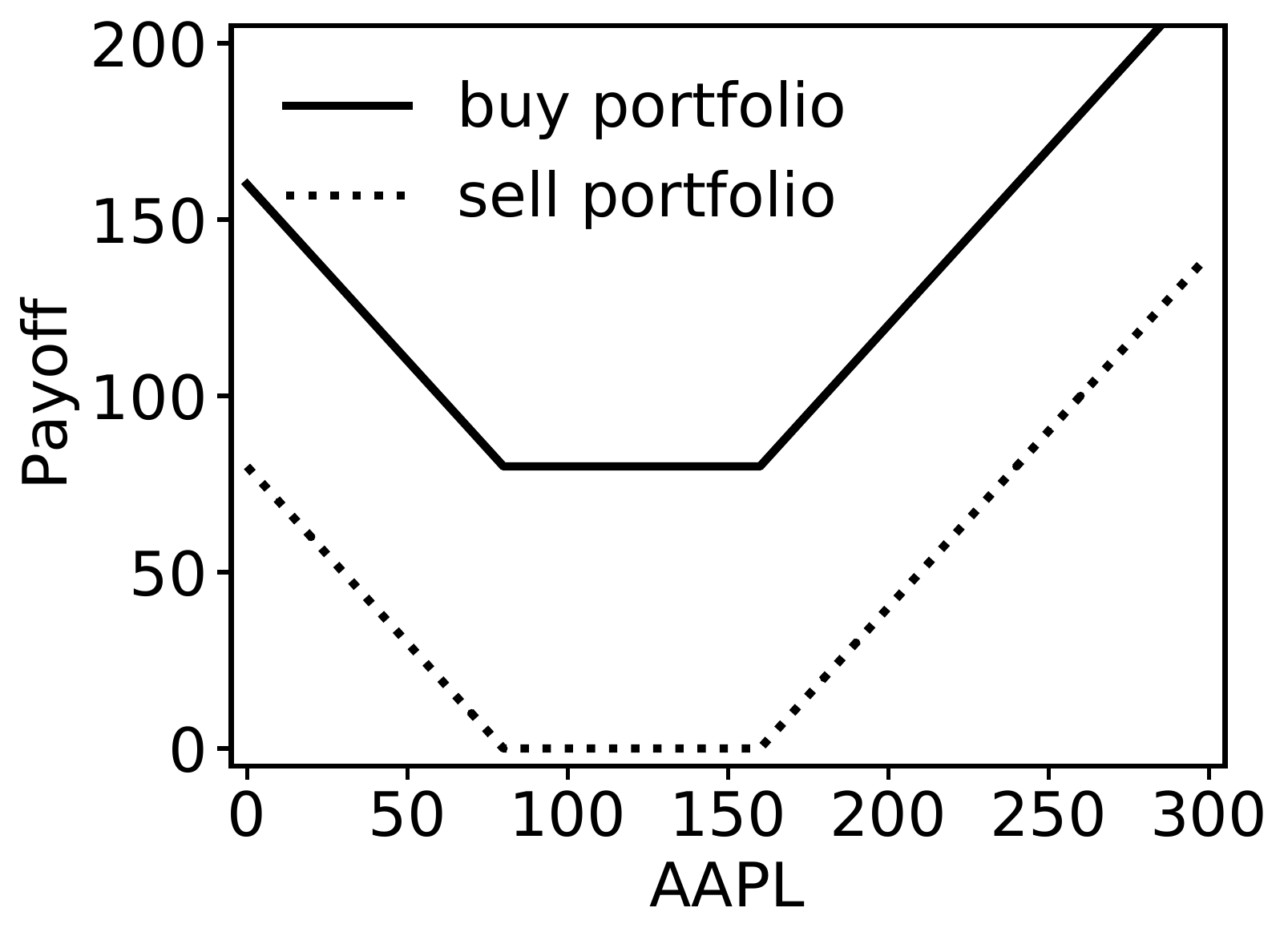}
		\caption{Payoffs of standard options bought and sold in \Ex{neg_L} as a function of AAPL value.}
		\label{subfig:lstzero}
	\end{subfigure}
	\caption[Payoffs of the matched options as a function of the value of the underlying asset at expiration.]{Payoffs of the matched options as a function of the value of the underlying asset at expiration. Fig.~\ref{subfig:lgtzero} shows the case of $L>0$, and Fig.~\ref{subfig:lstzero} the case of $L<0$.}
	\label{fig:arb_example}
\end{figure}

We now analyze the complexity of running \Mech{single_match}.
The left-hand side of constraint \eqref{eq:single_constraint} is a linear combination of max functions, and thus is a piecewise linear function of $S$.
Therefore, it suffices to solve \ref{mech:single_match} by satisfying constraints defined by $S$ at each breakpoint.
In our case, breakpoints of the constraint \eqref{eq:single_constraint} 
are the defined strike values in the market, plus two endpoints: $\vp \cup \vq \cup \{0, \infty\}$.
Let $n_K$ denotes the number of distinct strike values in the market, which is bounded above by 
$n_{\text{orders}} = M+N$, the total number of orders in the market.
Therefore, \ref{mech:single_match} is a linear program that has $n_K+2$ payoff constraints, and requires time polynomial in the size of the problem instance to solve.
We defer the complete proof to Appendix A.1.
\begin{theorem}
	\label{thm:consolidate_standard_options}
	\Mech{single_match} matches financial options written on the same underlying asset and expiration date across all types and strike prices in time polynomial in the number of orders.
\end{theorem}
\paragraph{Remarks.}
Several extensions can be directly applied to the mechanism:
\begin{enumerate}
    \setlength{\itemsep}{0pt}
    \item We can incorporate the time value of investments by multiplying $L$ by a (discount) rate in the objective of \ref{mech:single_match}.
    \item We note that matching solutions returned by \ref{mech:single_match} may involve fractional shares of stocks. Ideally, an exchange (e.g., cash-settled markets) allows it, and if an exchange (e.g., physical-settled markets) does not, we would need to normalize or round.
    \item Mechanism~\ref{mech:single_match} identifies a maximal bundle of orders that can be accepted without risk to the exchange, but does not specify what to do with the surplus if any. The surplus can be split arbitrarily among the involved traders and the exchange. We note that while the surplus is guaranteed to be nonnegative, it may have an uncontingent (cash) component and a state-contingent component.
\end{enumerate}

%
\subsection{Quoting Prices for Standard Financial Options}
A standard exchange maintains the best quotes (i.e., the highest bid and lowest ask) for each independent options market.
This section extends \Mech{single_match} to quote the most competitive prices for any \emph{custom} option of a specified type and strike by considering all other options related to the same underlying security.
We describe the price quote procedure in an arbitrage-free options market\footnote{That is, no match will be returned by \Mech{single_match}.} in regard to a single underlying asset $S$ with an expiration date $T$, represented by $(\vphi, \vp, \vb, \vpsi, \vq, \va)$. 
We defer the detailed proof of correctness to Appendix A.2. 

\begin{enumerate}
	\setlength{\itemsep}{8pt}
	\item The best bid $b^*$ for an option $(\chi, S, K)$ is the maximum gain of selling a portfolio of options that is \emph{weakly dominated} by $(\chi, S, K)$ with some constant offset $L$.
	
	We derive $b^*$ by adding the option $(\chi, S, K)$ to the sell side of the market indexed $N+1$ (as the exchange buys from sell orders), initializing its price $a_{N+1}$ to 0, and solving for \ref{mech:single_match}. 
	The best bid $b^*$ is then the returned objective.
	
	\item The best ask $a^*$ for an option $(\chi, S, K)$ is the minimum cost of buying a portfolio of options that \emph{weakly dominates} $(\chi, S, K)$ with some constant offset $L$.
	
	We derive $a^*$ by adding the option $(\chi, S, K)$ to the buy side of the market indexed $M+1$, initializing its price $b_{M+1}$ to a large number (e.g., $10^6$), and solving for \ref{mech:single_match}. 
	The best ask $a^*$ is then $b_{M+1}$ minus the returned objective.
\end{enumerate}
In the case of matching orders with multiple units, it is necessary to consider all orders in the market.
For deciding the existence of a match or quoting (instantaneous) prices, however, we only need to consider a set of orders that have the most competitive prices.
We define a set with such orders as a frontier set $\F$.
\begin{definition} [A Frontier Set of Options Orders]
	\label{def:opt_frontier}
	An option order is in the \emph{frontier set} $\F$ if its bid price is greater than or equal to the maximum gain of selling a weakly dominated portfolio of options with some offset $L$, or if its ask price is less than or equal to the minimum cost of buying a weakly dominant portfolio of options for some offset $L$.
\end{definition}

\begin{corollary}
	\label{coro:frontier_set}
	\Mech{single_match} determines the existence of a match and returns the instantaneous price quote in time polynomial in $\card{\F}$.
\end{corollary}

The proof of Corollary~\ref{coro:frontier_set} is deferred to Appendix A.3, 
which shows that in order to determine the existence of a match or quote the most competitive prices (i.e., the highest bid and the lowest ask) for any target option $(\chi, S, K)$, it suffices to consider options orders in $\F$ and run \Mech{single_match}.
The runtime complexity follows immediately from \Thm{consolidate_standard_options}.

\section{Combinatorial Financial Options}
\label{sec:combo_options}
This section introduces \emph{combinatorial financial options} and designs a market to trade such options. 
Combinatorial options extend standard options to more general derivative contracts that can be written on any linear combination of $U$ underlying assets.
We formally define a combinatorial financial option and its specifications.
\vspace{1ex}
\begin{definition} [Combinatorial Financial Options]
	Consider a set of $~U$ underlying assets.
	Combinatorial financial options are derivative contracts that specify the right to buy or sell a linear combination of the $U$ underlying assets at a strike price on an expiration date.
	Each contract specifies a call or put type $\chi \in \{1, -1\}$, a weight vector $\vomega \in \R^{U}$, a strike price $K \geq 0$, and an expiration date $T$.
	It has a payoff of $\max\{\chi(\vomega^\top \S-K), 0\}$, where $\S \in \R^U_{\geq 0}$ is a vector of the underlying assets' values at $T$.
\end{definition}
\vspace{2ex}
%
Consider a combinatorial option $C(\MSFT-\AAPL, 0)$ that has weight $1$ for \MSFT, weight $-1$ for \AAPL, and a strike price of zero.
An investor who buys the option bets on the event that Microsoft outperforms Apple Inc., and will exercise it if $S_\MSFT > S_\AAPL$. 
Thus, unlike standard options that will pay off due to price changes of a single security, combinatorial options bet on relative movements between assets or groups of assets, thus enabling the expression of future correlations among different underlying assets. 

We note several distinctions and interpretations in regard to the definition of combinatorial financial options:
\begin{enumerate}
    \setlength{\itemsep}{6pt}
    
    \item Distinction between a combinatorial option and a combination of standard options. 
    
    Consider the above combinatorial option $C(\MSFT-\AAPL, 0)$ and a combination of two standard options $C(\MSFT, K)$ and $P(\AAPL, K)$. One would exercise the combinatorial option and possess the portfolio, i.e., long a share of \MSFT and short sell a share of \AAPL, as long as $S_\MSFT \geq S_\AAPL$, whereas in the combination of standard options, one would exercise both and possess the same portfolio, if $S_\MSFT\geq K$ \emph{and} $S_\AAPL \leq K$.%
    \footnote{In terms of possessing the portfolio \MSFT-\AAPL, $C(\MSFT-\AAPL, 0)$ captures all combination of two standard options of the form, $C(\MSFT, K)$ and $P(\AAPL, K)$ for all positive $K$.}
    Thus, the exercise of a combinatorial option does not imply a simultaneous exercise of all standard options in the combination, and vice versa, the simultaneous buy and sell of several standard financial options with different underlying securities and strike prices (e.g., multi-leg options orders) cannot replicate the payoff of a combinatorial option.

    \item Distinction between a call and a put combinatorial option.
    
    We note that the distinction between a call and a put for combinatorial options depends on the strike price and coefficients that one specifies in a contract.
    For instance, in the above example, $C(\MSFT-\AAPL, 0)$ is identical to $P(\AAPL-\MSFT, 0)$, as they have the same payoff function $\max\{S_\MSFT-S_\AAPL, 0\}$.
    Despite the different interpretations and expressions, we follow the convention of standard financial options, and have the strike price always be non-negative.
\end{enumerate}


\subsection{Matching Orders on Combinatorial Financial Options}
The increased expressiveness in combinatorial options brings new challenges in market design: 
only matching buy and sell orders on options related to the same assets and weights may yield few or no trades, despite plenty of profitable trades among options written on different portfolios.
We start by giving the following motivating examples to illustrate such scenarios.
\begin{example}[Matching combinatorial option orders]
	\label{ex:combo_options}
	Consider a combinatorial options market with four orders
	\begin{itemize}
		\setlength{\itemsep}{2pt}
		\item $o_1$: buy one $C(1\AAPL+2\MSFT, 300)$ at bid \$110;
		\item $o_2$: buy one $C(1\AAPL+1\MSFT, 300)$ at bid \$70;
		\item $o_3$: sell one $C(1\AAPL+3\MSFT, 300)$ at ask \$160;
		\item $o_4$: sell one $C(1\AAPL, 250)$ at ask \$5.
	\end{itemize}
	The exchange returns no match if it only considers options related to the same combination of assets.
	However, a profitable match does exist.
	The exchange can sell to $o_1$ and $o_2$ and simultaneously buy from $o_3$ and $o_4$ to get an immediate gain of \$15 (110+70-160-5).
	Figure~\ref{fig:combo_match_example} plots the overall payoff (which is always non-negative), as a function of $S_\AAPL$ and $S_\MSFT$ on the expiration date:
	{
		\begin{align}
		\Psi := &\max\{S_\AAPL+3S_\MSFT-300, 0\} + \max\{S_\AAPL-250, 0\} - \notag\\
		&\max\{S_\AAPL+2S_\MSFT-300, 0\} - \max\{S_\AAPL+S_\MSFT-300, 0\},\notag
		\end{align}
	}
	Therefore, the exercises of options cannot subtract from the \$15 immediate gain, but could add to it, depending on the future prices of the two stocks. 
	\qed
\end{example}

\begin{figure}[h]
	\centering
	\includegraphics[width=0.45\columnwidth]{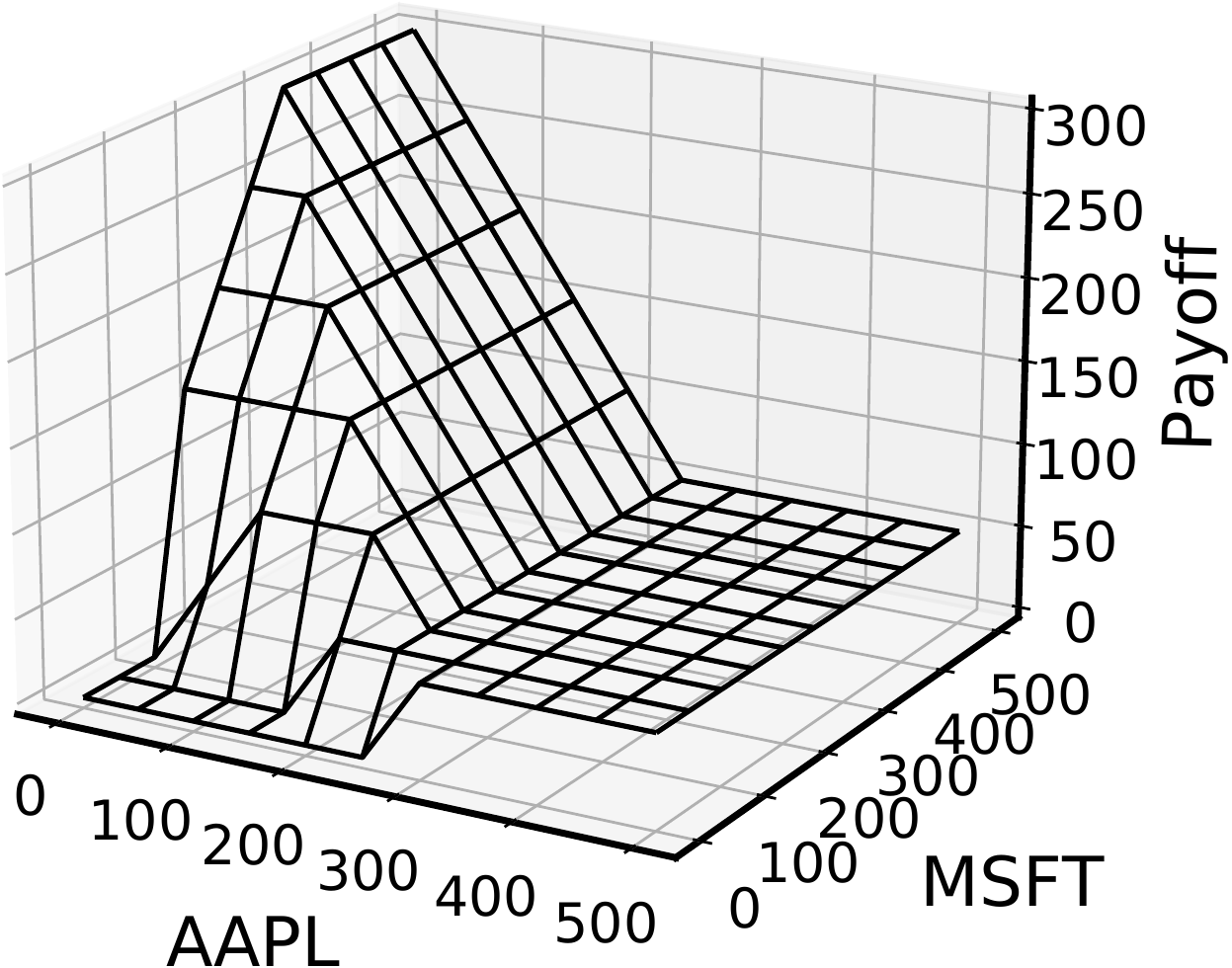}
	\caption[Payoff of combinatorial options matched in \Ex{combo_options} as a function of $S_\AAPL$ and $S_\MSFT$.]{Payoff of combinatorial options matched in \Ex{combo_options} as a function of $S_\AAPL$ and $S_\MSFT$. The example demonstrates the case of $L=0$.}
	\label{fig:combo_match_example}
\end{figure}

In the above example, the exchange can consider matching each individual buy order sequentially to some combination of sell orders.
For example, the exchange can first sell to buy order $o_1$ and buy from sell orders $\frac{2}{3}o_3$ and $\frac{1}{3}o_4$, and then sell to buy order $o_2$ and buy from sell orders $\frac{1}{3}o_3$ and $\frac{2}{3}o_4$.
Both are profitable trades subject to no loss, leading to the same ultimate match as in Example~\ref{ex:combo_options}.
The next example shows that such sequential matching of each individual buy order to a combination of sell orders may fail to find valid trades.
\begin{example}[Matching combinatorial option orders in batch]
	\label{ex:combo_options_2}
	Consider the following four combinatorial options orders
	\begin{itemize}
		\setlength{\itemsep}{2pt}
		\item $o_1$: buy one $C(\text{A}+\text{B}, 10)$ at bid \$6;
		\item $o_2$: buy one $C(\text{B}+\text{C}, 7)$ at bid \$6;
		\item $o_3$: sell one $C(\text{A}+\text{B}+\text{C}, 7)$ at ask \$10;
		\item $o_4$: sell one $C(\text{B}, 3)$ at ask \$2.
	\end{itemize}
	No match will be found if we consider each buy order individually: covering a sold combinatorial option in $o_1$ or $o_2$ requires buying the same fraction of $o_3$, which is at a higher price and will incur a net loss.
	However, a valid match does exist by selling to $o_1$ and $o_2$ and buying from $o_3$ and $o_4$.
	It costs the exchange \$0, and can yield a positive payoff in the future. 
	The exchange has
	{
		\begin{align*}
		\max\{S_A+S_B-10, 0\} + \max\{S_B+S_C-7, 0\}
		\leq \max\{S_A+S_B+S_C-7, 0\} + \max\{S_B-3, 0\},
		\end{align*}
	}
	meaning the liability will always be no larger than the payoff of bought options, for all non-negative $S_A, S_B, S_C$.%
	\footnote{To see this, we can add $\max\{S_A-7, 0\}$ on both sides of the inequality, and have $\max\{S_A-7, 0\} + \max\{S_B+S_C-7, 0\} \leq \max\{S_A+S_B+S_C-7, 0\}$ and $\max\{S_A+S_B-10, 0\} \leq \max\{S_A-7, 0\} + \max\{S_B-3, 0\}$ by Jensen's inequality.}
	\qed
\end{example}

We now introduce the formal setting of a combinatorial financial options market.
A combinatorial options market is a two-sided market with a set of buy limit orders, indexed by $m \in \{1, 2, ..., M\}$, and a set of sell limit orders, indexed by $n \in \{1, 2, ..., N\}$. 
Buy orders are represented by a type vector $\vphi \in \{1, -1\}^M$ with each entry specifying a put or a call combinatorial option, a weight matrix $\valpha \in \R^{U \times M}$ specifying the linear combinations, a strike vector $\vp \in \R_{\geq 0}^M$, and a bid-price vector $\vb \in \R_{+}^M$.
Sell orders are defined by a separate type vector $\vpsi \in \{1, -1\}^N$, a weight matrix $\vbeta \in \R^{U \times N}$, a strike vector $\vq \in \R_{\geq 0}^N$, and an ask-price vector $\va \in \R_{+}^N$.
Similar to a standard options market, the exchange decides the fraction $\vgamma \in [0, 1]^M$ to sell to buy orders and the fraction $\vdelta \in [0, 1]^{N}$ to buy from sell orders to maximize net profit.
We generalize the proposed mechanism \ref{mech:single_match} for standard options to facilitate trading combinatorial options: 
%
\begin{align}\tag{M.2}
\label{mech:combo_match}
\max \limits_{\vgamma, \vdelta, L} & \displaystyle \quad \vb^\top \vgamma - \va^\top \vdelta - L\\
\text{s.t.} & \displaystyle \quad \displaystyle \sum_{m} \gamma_m \max\{\phi_m(\valpham^\top \S - p_m), 0\} -   \sum_{n} \delta_n \max\{\psi_n(\vbetan^\top \S - q_n), 0\} \leq L \quad \forall \S \in \R_{\geq 0}^U \label{eq:constraint}
\end{align}

However, unlike \ref{mech:single_match}, it is no longer feasible to solve the optimization problem \ref{mech:combo_match} by enumerating the constraint at each breakpoint, which requires iterating over every combination of underlying asset values.
The number of constraints can grow exponentially as $\order(2^{M+N})$ or $\order((M+N)^U)$ by Sauer's Lemma.\footnote{In fact, we can write \ref{mech:combo_match} as an exponential-sized linear program.}

We analyze the complexity of finding the optimal match in a combinatorial options market, i.e., solving for \ref{mech:combo_match}.
We first show in the following theorem that given a market instance, it is NP-complete to decide if a certain matching assignment, $\vgamma$ and $\vdelta$, violates the constraint \eqref{eq:constraint} for a fixed $L$. 
%
\begin{theorem}
	\label{thm:NP-complete}
	Consider all combinatorial options in the market $(\vphi, \valpha, \vp, \vpsi, \vbeta, \vq)$. 
	For any fixed $L$, it is NP-complete to decide 
	\begin{itemize}
		\setlength{\itemsep}{5pt}
		\item Yes: $\vgamma = \vdelta = \vone$ violates the constraint in \ref{mech:combo_match} for some $\S$,
		\item No: $\vgamma = \vdelta = \vone$ satisfies the constraint in \ref{mech:combo_match} for all $\S$,
	\end{itemize}
	even assuming that each combinatorial option is written on at most two underlying assets.
\end{theorem}
\begin{proof}
	The decision problem is in NP.
	Given a certificate which is a value vector $\S \in \R^U_{+}$, we plug $\S$ into the constraint of \ref{mech:combo_match} to compute the payoff and check whether it is less than $L$.
	This takes time $\order(U(M+N))$.
	
	For the NP-hardness, we prove by reducing from the Vertex Cover problem --- given an undirected graph $G=(V,E)$ and an integer $k$, decide if there is a subset of vertices $V' \subseteq V$ of size $k$ such that each edge has at least one vertex in $V'$.	
		Given a Vertex Cover instance $(G, k)$, we construct an instance of the combinatorial options matching problem.
		Let the set of underlying assets correspond to vertices in $G$, i.e., $U=|V|$.
		For each vertex indexed $i$, we associate four options with it, one on the buy side and three on the sell side. They have payoff functions as follows:
		\begin{align*}
		f_i &=\max\{2K_1S_i-K_1, 0\}, &g^{(1)}_i &=\max\{K_1S_i, 0\},\\
		g^{(2)}_i &=\max\{K_2S_i-K_2, 0\}, &g^{(3)}_i &=\max\{S_i, 0\},
		\end{align*}
		where we choose $K_1$ and $K_2$ for some large numbers with $K_2 \gg K_1$.
		For example, we choose $K_1=10|E|$ and $K_2=100|E|$.
		For each edge $e = (i, j)$, we define two options that involve its two end-points $i$ and $j$, one on the buy side and one on the sell side.
		They have payoff functions:
		\begin{align*}
		f_e &= \max\{S_i+S_j, 0\}, &g_e = \max\{S_i+S_j-1, 0\}.
		\end{align*}
		Finally, we include one sell order on an option with payoff $g^\star = \max\{|E|-k-L-0.5,0\}$. %
		Since $L$ is fixed in advance, we assume $|E|-k-L-0.5>0$ without loss of generality.
		Thus, we have $M=|V|+|E|$ buy orders and $N=3|V|+|E|+1$ sell orders.
		The construction takes time polynomial in the size of the Vertex Cover instance.
		
		\emph{Suppose the Vertex Cover instance is a Yes instance.}
		We assume that $\{v_1, v_2, ..., v_k\}$ is a vertex cover.
		We show that assigning $S_1,S_2,...,S_k$ to 1 for the selected underlying assets (i.e., vertices) and 0 for the rest unselected ones gives an $\S$ that violates the constraint \eqref{eq:constraint}.
		The left-hand side of the constraint \eqref{eq:constraint} is
		\begin{align}
		\label{eq:LHS}
		z := & \,\,\underbrace{\!\!\sum_{i \in V}\left(f_i-g^{(1)}_i-g^{(2)}_i-g^{(3)}_i\right)\!\!}_{z_v}\,\,+\,\,\underbrace{\!\!\sum_{e \in E} (f_e-g_e)\!\!}_{z_e}\,\,-g^\star.
		\end{align}
		For $S_i\in\{0,1\}$, it is easy to see that $f_i-g_i^{(1)}-g_i^{(2)}=0$.
		Thus, we have $z_v = -\sum_{i \in V}g^{(3)} = -k$ by our assignment.
		Since at least one of $S_i,S_j$ is $1$ for any edge $(i,j)\in E$, we have $f_e - g_e = 1$ and $z_e = \card{E}$.
		Therefore, we have 
		\begin{align*}
		z = & -k + |E| - (|E| - k - L - 0.5) = L + 0.5 > L.
		\end{align*}
		
		\emph{Suppose the Vertex Cover instance is a No instance.}
		We aim to show that for the given $\vgamma$ and $\vdelta$, there does not exist an $\S$ that violates the constraint.
		We prove by maximizing $z$ and demonstrating $z \leq L$.	
		We start by proving the following claim.
		\begin{claim}
			\label{thm:claim1}
			For an optimal $z$, we have $S_i \in \{0,1\}$.
		\end{claim}
		We prove by contradiction. First assume $S_j > 1$ for some $j$.
		We rewrite Eq.~\eqref{eq:LHS} and have
		\begin{align*}
		z := & \,\,\underbrace{\!\!f_j-g^{(1)}_j-g^{(2)}_j-g^{(3)}_j\!\!}_{z_j}\,\, + \,\,\underbrace{\!\!\sum_{i\in V\setminus j}\left(f_i-g^{(1)}_i-g^{(2)}_i-g^{(3)}_i\right)\!\!}_{z_i}\,\, + \,\,\underbrace{\!\!\sum_{e \in E} (f_e-g_e)\!\!}_{z_e}\,\,-g^\star.
		\end{align*}
		We first analyze $z_j$ and have
		\begin{align*}
		z_j = & \max\{2K_1 S_j-K_1, 0\} - \max\{K_1S_j, 0\} - \max\{K_2 S_j-K_2, 0\} - \max\{S_j, 0\}\\
		= & 2K_1 S_j-K_1 - K_1S_j - (K_2S_j-K_2) - S_j \tag{by assumption of $S_j > 1$}\\
		= & K_2 - K_1 - (K_2-K_1+1)S_j.
		\end{align*}
		Recall that we choose $K_2 \gg K_1$, e.g., $K_1=10|E|$ and $K_2=100|E|$.
		Thus, we have $z_j$ increase with rate $K_2-K_1+1$, as $S_j$ decreases uniformly.
		Since $z_e$ decreases at most $|E|$ and the rest two terms, $z_i$ and $g^\star$, do not depend on $S_j$, decreasing $S_j$ increases $z$.
		It is sub-optimal to have $S_j > 1$ for some $j$.
		
		Next, we assume $0 \leq S_j \leq 1$ for some $j$ and have
		\begin{align*}
		z_j 
		&= \begin{cases}\displaystyle
		-K_1S_j-S_j &\quad 0 \leq S_j \leq 0.5,\\
		\displaystyle
		K_1(S_j-1)-S_j &\quad 0.5 < S_j \leq 1.\\
		\end{cases}
		\end{align*}
		As $K_1$ is large, by a similar argument analyzing the growth rate of each term, we show that $z_j$ (and also $z$) increases by assigning $S_j$ to $0$ if $0 \leq S_j \leq 0.5$ and by assigning $S_j$ to $1$ if $0.5 < S_j \leq 1$.
		
		We finish proving Claim~\ref{thm:claim1}, and now have $S_i \in \{0,1\}$.
		Following Eq.~\eqref{eq:LHS}, we have
		\[z = -\sum_{i \in V}S_i + \sum_{e \in E} (f_e-g_e) - (|E|-k)+L+0.5.\]
		The goal is to maximize $z$ and show $z\leq L$.
		As $S_i$ is an integer, to show $z\leq L$, it suffices to show that $\sum_{e \in E} (f_e-g_e)-\sum_{i \in V}S_i<|E|-k$.
		We prove by contradiction, assuming $\sum_{e \in E} (f_e-g_e)-\sum_{i \in V}S_i \geq |E|-k$.
		Recall that to have $f_e-g_e = 1$ for $e=(i,j)$, we need at least one of $S_i,S_j$ to be $1$.
		We consider the following two possible cases, and aim to refute them:
		\begin{enumerate}[(a)]
			\setlength{\itemsep}{5pt}
			\item $\sum_{e \in E} (f_e-g_e)=|E|$ and $\sum_{i \in V}S_i \leq k$.
			
			This means we cover all edges with at most $k$ vertices assigned to 1.
			
			\item $\sum_{e \in E} (f_e-g_e)<|E|$ and $\sum_{i \in V}S_i < k$.
			
			For any $e$ with $f_e-g_e=0$, we can assign 1 to one of its end-points, and have $\sum_{e \in E} (f_e-g_e)=|E|$ without changing $\sum_{e \in E} (f_e-g_e)-\sum_{i \in V}S_i$. This leads back to (a).
		\end{enumerate}
		Both cases contradict to the fact that the Vertex Cover instance is a No instance.
		We have $\sum_{e \in E} (f_e-g_e)-\sum_{i \in V}S_i < |E|-k$ as desired, and thus $z \leq L$ for all $\S$.
\end{proof}
\vspace{1ex}

Using a slightly stronger version of \Thm{NP-complete} (see Theorem 5 in Appendix B.1), 
we show that optimal clearing of a combinatorial options market is coNP-hard.
We defer the proof of \Thm{coNP-hard} to Appendix B.1.
\begin{theorem}
	\label{thm:coNP-hard}
	Optimal clearing of a combinatorial options market $(\vphi, \valpha, \vp, \vb, \vpsi, \vbeta, \vq, \va)$ is coNP-hard, even assuming that each combinatorial option is written on at most two underlying assets.
\end{theorem}
\subsection{A Constraint Generation Algorithm to Match Combinatorial Financial Options}
Since it is no longer practical to solve \ref{mech:combo_match} by identifying all constraints defined by different combinations of underlying asset values, we propose an algorithm (\Algo{comb_match}) that finds the exact optimal match through iterative constraint generation.
We first explain how \Algo{comb_match} works, and then prove that it is an equivalent formulation of \ref{mech:combo_match}.

At the core of \Algo{comb_match} is a constraint generation process, where a new optimization problem is defined per iteration to find a violated constraint.
The constraint set $\Cc$ includes the $\S \in \R^U_{\geq 0}$ value vectors that define different constraints. We can plug each $\S$ back into the constraint to define $\vf$ and $\vg$, which are the realized payoffs of bought and sold options with respect to the $\S$.
We start with the constraint set $\Cc$ of a zero vector, meaning that all underlying assets have price zero at expiration.
In each iteration, the upper-level optimization problem (\ref{eq:combo_upper}) computes the optimal match (i.e., $\vgamma^*$, $\vdelta^*$, and $L^*$) that satisfies this restrictive set of generated constraints.
Thus, it is a linear program with $\card{\Cc}$ constraints.
We then use the lower-level optimization problem (\ref{eq:comb_milp}) to find the $\S^*$ value vector that violates the returned matching solution the most.
That is, given $\vgamma^*$, $\vdelta^*$, and $L^*$, it generates the most adversarial realization of $\S$ that yields the worst-case payoff loss for the exchange at expiration.
We show that this lower-level optimization can be formulated as a mixed-integer linear program.
We then include this generated $\S^*$ vector into the constraint set.

The exact optimal match is returned when the lower-level \ref{eq:comb_milp} gives an objective value of zero, meaning there exists no such $\S$ that violates the upper-level matching assignment.
Therefore, the final number of constraints in $\Cc$ is the number of iterations that takes \Algo{comb_match} to terminate.
The algorithm trivially terminates finitely, but similar to the simplex method, it has no guarantee on the rate of convergence.
We later evaluate its performance on synthetic combinatorial options markets of different scales, and demonstrate that \Algo{comb_match} converges relatively fast, with the number of iterations increasing linearly in the size of a market instance.

We next prove \Algo{comb_match} returns the same optimal clearing solution as \ref{mech:combo_match}.
We first show in the Lemma that \ref{eq:comb_milp} finds the value vector $\S$ that violates the constraint \eqref{eq:constraint} in \ref{mech:combo_match} the most.
\begin{algorithm}[t]
	\caption{Match orders in a combinatorial options market.}
	\label{algo:comb_match}
	\begin{algorithmic}[1]
		\Statex \textbf{Input:}~%
		A combinatorial options market defined by
		$(\vphi, \valpha, \vp, \vb, \vpsi, \vbeta, \vq, \va)$.
		\Statex \textbf{Output:}~%
		An optimal clearing that matches $\vgamma^*$ buy orders to $\vdelta^*$ 
		sell orders.
		\medskip
		\State Initialize $z \gets \infty$, $\S \gets \vzero$, 
		\Statex ~\hphantom{Initializ} $\vf \gets \max\{\vphi(\valpha^\top \S-\vp), \vzero\}$, $\vg \gets \max\{\vpsi(\vbeta^\top \S-\vq), \vzero\}$, $\Cc \gets \{(\vf, \vg)\}$.
		\While {$z > 0$}
		\State Solve the following upper level LP and get the optimal $(\vgamma^*, \vdelta^*, L^*)$
		\begin{align*}
		\label{eq:combo_upper}
		&\max \limits_{\vgamma, \vdelta, L}  \quad \vb^\top \vgamma - \va^\top \vdelta - L \tag{M.3U}\\
		&\quad\text{s.t.} \quad \vgamma^\top \vf - \vdelta^\top \vg \leq L \quad &\forall (\vf, \vg) \in \Cc
		\end{align*}
		\State Given $(\vgamma^*, \vdelta^*, L^*)$, solve the following lower level MILP  and get the optimal 
		\Statex ~\hphantom{w } $(\S^*, \vf^*, \vg^*, \vI^*, z^*)$ 
		\begin{align}
			\label{eq:comb_milp}
			& \max \limits_{\S, \vf, \vg, \vI} \quad z := \vgamma^\top \vf - \vdelta^\top \vg - L \tag{M.3L}\\
			& \text{  s.t.} \quad \quad \phi_m(\valpham^\top \S - p_m) \geq \M(\I_m-1) \notag\\
			& \hphantom{\text{  s.t.}} \quad \quad \phi_m(\valpham^\top \S - p_m) \leq \M \I_m \notag\\
			& \hphantom{\text{  s.t.}} \quad \quad f_m \leq \phi_m(\valpham^\top \S - p_m) - \M(\I_m-1) \notag\\
			& \hphantom{\text{  s.t.}} \quad \quad f_m \leq \M \I_m \notag\\
			& \hphantom{\text{  s.t.}} \quad \quad \I_m \in \{0, 1\} \qquad \qquad \forall m \in \{1,...,M\} \notag\\
			& \hphantom{\text{  s.t.}} \quad \quad g_n \geq \psi_n(\vbetan^\top \S - q_n) \notag\\
			& \hphantom{\text{  s.t.}} \quad \quad g_n \geq 0 \qquad \qquad \qquad \forall n  \in \{1,...,N\} \notag
 		\end{align}
		\State $\Cc \gets \Cc \cup (\vf^*, \vg^*)$, $z\gets z^*$
		\EndWhile
		\State \Return $\vgamma^*$ and $\vdelta^*$
	\end{algorithmic}
\end{algorithm}
\begin{lemma}
	Given fixed $\vgamma$, $\vdelta$, and $L$ for a combinatorial options market $(\vphi, \valpha, \vp, \vb, \vpsi, \vbeta, \vq, \va)$, \ref{eq:comb_milp} returns the value of underlying assets $\S$ that violates the constraint of \ref{mech:combo_match} the most.
\end{lemma}
\begin{proof}
	First, it is easy to see that the formulation below returns the $\S$ that violates constraint \eqref{eq:constraint} the most, since we will have the largest feasible $\vf$ and the smallest feasible $\vg$ at the optimum.
	That is, $ f_m = \max\{\phi_m(\valpham^\top \S - p_m), 0\}$ and $g_n = \max\{\psi_n(\vbetan^\top \S - q_n), 0\}$.
		\begin{align}
		& \max \limits_{\S, \vf, \vg} \quad \vgamma^\top \vf - \vdelta^\top \vg - L \notag\\
		& \text{  s.t.} \quad \quad f_m \leq \max\{\phi_m(\valpham^\top \S - p_m), 0\} &\forall m \in \{1,...,M\} \notag\\
		& \phantom{\text{  s.t.}} \quad \quad g_n \geq \psi_n(\vbetan^\top \S - q_n) \notag\\
		& \phantom{\text{  s.t.}} \quad \quad g_n \geq 0 &\forall n  \in \{1,...,N\} \notag	
		\end{align}
	
	It remains to show the set of constraints related to any buy order $m$ in \ref{eq:comb_milp} is equivalent to $f_m \leq \max\{\phi_m(\valpham^\top \S - p_m), 0\}$.
	In short, the set of constraints in \ref{eq:comb_milp} linearize the max functions by using the big-$\M$ trick on each binary decision variable $\I_m$, where $\M$ is a large constant, say $10^6$.
	
	Consider each case of $\I_m \in \{0, 1\}$.
	The first two constraints, $\phi_m(\valpham^\top \S - p_m) \geq \M(\I_m-1)$ and $\phi_m(\valpham^\top \S - p_m) \leq \M \I_m$, guarantee that $\phi_m(\valpham^\top \S - p_m) \geq 0 \iff \I_m=1$.
	Then, following the third and fourth constraints, i.e., $f_m \leq \phi_m(\valpham^\top \S - p_m) - \M(\I_m-1)$ and $f_m \leq \M \I_m$, we have 
	\[f_m \leq 0 \quad \text{if $\I_m = 0$}; \quad f_m \leq \phi_m(\valpham^\top \S - p_m) \quad \text{if $\I_m = 1$.}\]
	Therefore, at the optimum, we have 
	\[f_m = 0 \quad \text{if $\I_m = 0$}; \quad f_m = \phi_m(\valpham^\top \S - p_m) \quad \text{if $\I_m = 1$,}\]
	which is equivalent to $f_m = \max\{\phi_m(\valpham^\top \S - p_m), 0\}$.
\end{proof}
Therefore, when \ref{eq:comb_milp} returns an objective value of zero, the constraint \eqref{eq:constraint} in \ref{mech:combo_match} is satisfied for all $\S$, and \Algo{comb_match} returns a valid match that optimizes for overall profit. 
\begin{theorem}
	\label{thm:equivalent_form}
	Given a combinatorial options market instance $(\vphi, \valpha, \vp, \vb, \vpsi, \vbeta, \vq, \va)$, \Algo{comb_match} returns the optimal clearing defined in \ref{mech:combo_match}.
\end{theorem}

\section{Experiments: Real-Market Standard Financial Options}
\label{sec:exp_standard_option}
We first evaluate the proposed mechanism~\ref{mech:single_match} that matches orders on standard options across types and strike values.
We conduct empirical analysis on the OptionMetrics dataset provided by the Wharton Research Data Services (WRDS), which contains real-market option prices, i.e., the best bid and ask prices, for each market defined by an underlying asset, an option type, a strike price, and an expiration date.%
\footnote{Our data includes American options that allow exercise before expiration. 
	In practice, American options are 
	almost always more profitable to sell than to exercise early \cite{Singh2019}.
	In experiments, we ignore early exercise and treat them as European options.}
We choose options data on 30 stocks that compose the DJI, as these stocks have actively traded options that cover a wide range of moneyness and maturity levels.
There are a total of 25,502 distinct options markets for the 30 stocks in DJI on January 23, 2019,
covering around 12 expiration dates for each stock.

We use \ref{mech:single_match} to consolidate the outstanding buy orders and sell orders from independently-traded options markets that associated with the same security and expiration date.
This reduces the original 25,502 distinct markets to a total of 366 markets, which have standard options across different types and strikes.
We run \ref{mech:single_match} on these consolidated markets to
	\begin{enumerate}[(1)]
		\setlength{\itemsep}{2pt}
		\item find trades that the current independent-market design cannot match,
		\item compute new price quotes implied by this consolidated market design, and
		\item compare matches and price quotes to those of the case when we restrict $L$ to 0.
	\end{enumerate}

Out of the 366 consolidated options markets, we find profitable matches in 150 markets, which failed to transact under the independent market design.
Among those trades, 94 have non-negative $L$ (i.e., the case of \Ex{nonneg_L}), making an average net profit of \$1.03, with a maximum of \$9.64; 
the remaining 56 have negative $L$ (i.e., the case of \Ex{neg_L}),  implying an average interest rate of 0.7\%, with a maximum at 2.02\%.
When we restrict $L$ to 0 (i.e., $L$ is no longer a decision variable), we are able to find matches in 74 markets.
Detailed statistics for options of each stock are available in the Appendix C. 

For the arbitrage-free consolidated markets (i.e., markets with no match returned by \ref{mech:single_match}), we find that approximately 49\% of the orders belong to the frontier set.
Using these orders to derive the most competitive bids and asks, we find that the bid-ask spreads can be reduced by 73\%, from an average of 80 cents for each option series in the independent markets to 21 cents in consolidated options markets.
For the case of $L$ set to 0, the bid-ask spreads can be reduced by 52\%.

These results show that the exchange, by consolidating options markets across types and strike prices, can potentially achieve a higher economic efficiency, matching orders that the current independent design cannot and providing more competitive bid and ask prices.

\section{Experiments: Synthetic Combinatorial Options Market}
\label{sec:exp_combo_option}
Since there is no combinatorial option traded in financial markets, we evaluate the proposed Algorithm~\ref{algo:comb_match} on synthetic combinatorial options, with prices calibrated using real-market standard options written on each related underlying security.%
\footnote{We adopt the same dataset as Section~\ref{sec:exp_standard_option}, and use standard options that expire on Febuary 1, 2019, to calibrate order prices for generated combinatorial options.}
We implement Algorithm~\ref{algo:comb_match} using Gurobi 9.0~\cite{gurobi}.
We are interested in quantifying the performance of \Algo{comb_match} in parametrically different markets that vary in the likelihood of matching, the number of orders, and the number of underlying assets.
We first describe our synthetic dataset.

\subsection{Generate Synthetic Orders}
We generate combinatorial options markets of $U$ underlying assets.
Each combinatorial option is written on a combination of two stocks, $S_i$ and $S_j$, randomly selected from the $U$ underlying assets.
This gives a total number of $U \choose 2$ asset pairs.
Weights for the selected assets, $w_i$ and $w_j$, are picked uniformly randomly from $\{\pm1, \pm2, \dotsc, \pm9\}$ and are preprocessed to be relatively prime.

We generate strikes and bid/ask prices using real-market standard options data related to each individual asset to realistically capture the value of the synthetic portfolio.
Let $\K_i$ and $\K_j$ respectively denote the set of strike prices offered by standard options on each selected asset.
We generate the strike $K$ by first sampling two strike values, $k_i \sim \K_i$ and $k_j \sim \K_j$, and scaling them by the associated weights to get $K = w_ik_i+w_jk_j$.
If $K$ is positive, we have a call option.
Instead, if $K$ is negative, we generate a put option, and update the strike to $-K$ and weights to $-w_i$ and $-w_j$ to comply with the representation and facilitate payoff computations.

We randomly assign each option to the buy or sell side of the market with equal probability, and generate a bid or ask price accordingly.
Similar to the price quote procedure in Section~\ref{sec:standard_option}, we derive the bid $b$ by calculating the maximum gain of selling a set of standard options whose payoff is dominated by the combinatorial option of interest, and compute the ask $a$ by calculating the minimum cost of buying a set of standard options whose payoff dominates the generated conbinatorial option.
We add noises around the derived prices to control the likelihood of matching in a market, and set final prices to $b(1+\zeta)$ or $a(1-\zeta)$, where $\zeta \sim [0, \eta]$ and $\eta$ is a noise parameter. 

\subsection{Evaluation}
\begin{figure}[t]
	\centering
	\begin{subfigure}{0.48\columnwidth}	
		\centering
		\includegraphics[width=0.9\columnwidth]{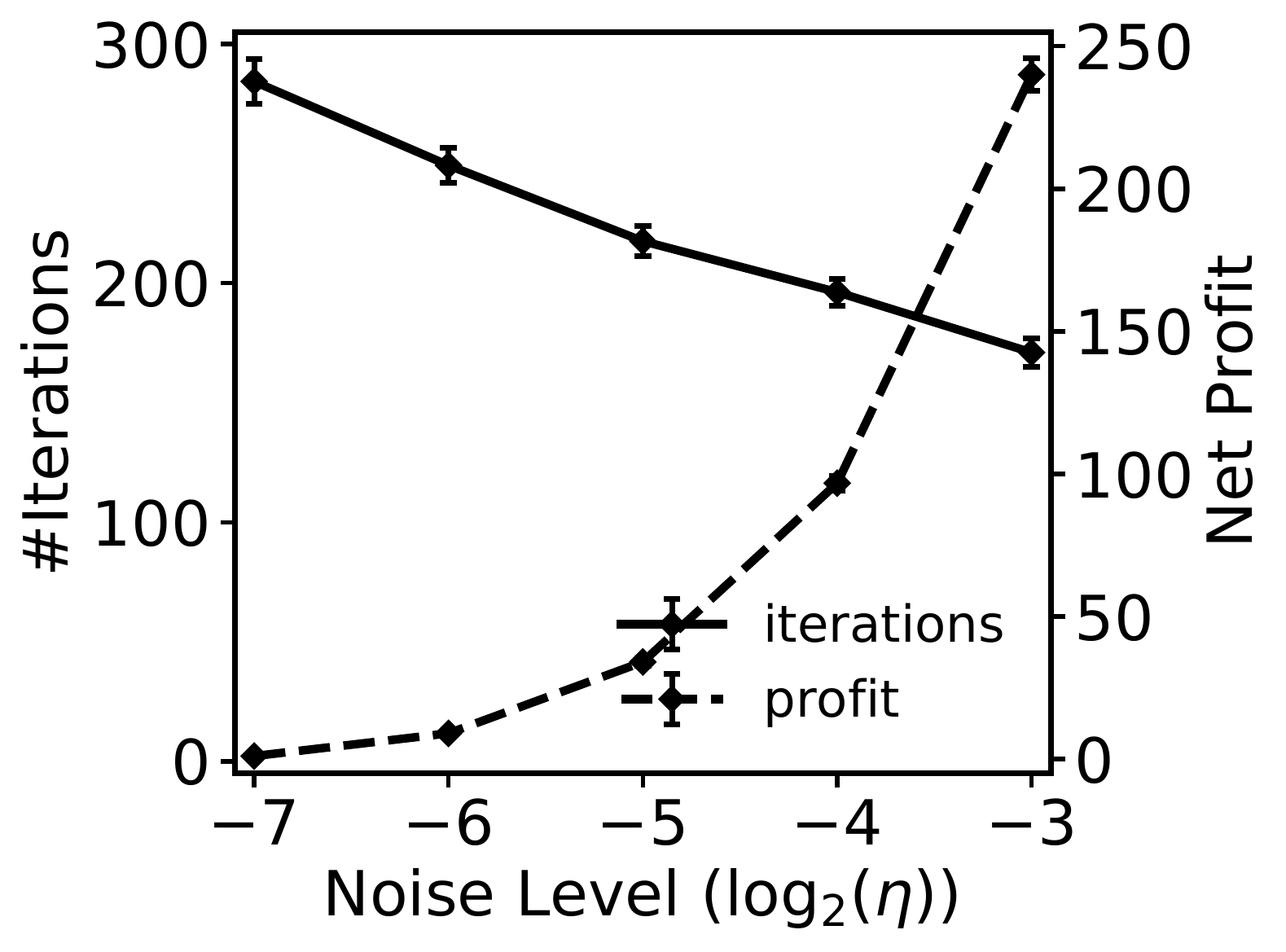}
		\caption{Vary noise $\eta$ added to order prices in markets with $U=4$ and $n_\text{orders}=150$.}
		\label{subfig:combo_plots_a}
		\vspace{2ex}
	\end{subfigure}
	\hspace{0.02\columnwidth}
	\begin{subfigure}{0.48\columnwidth}	
		\centering
		\includegraphics[width=0.9\columnwidth]{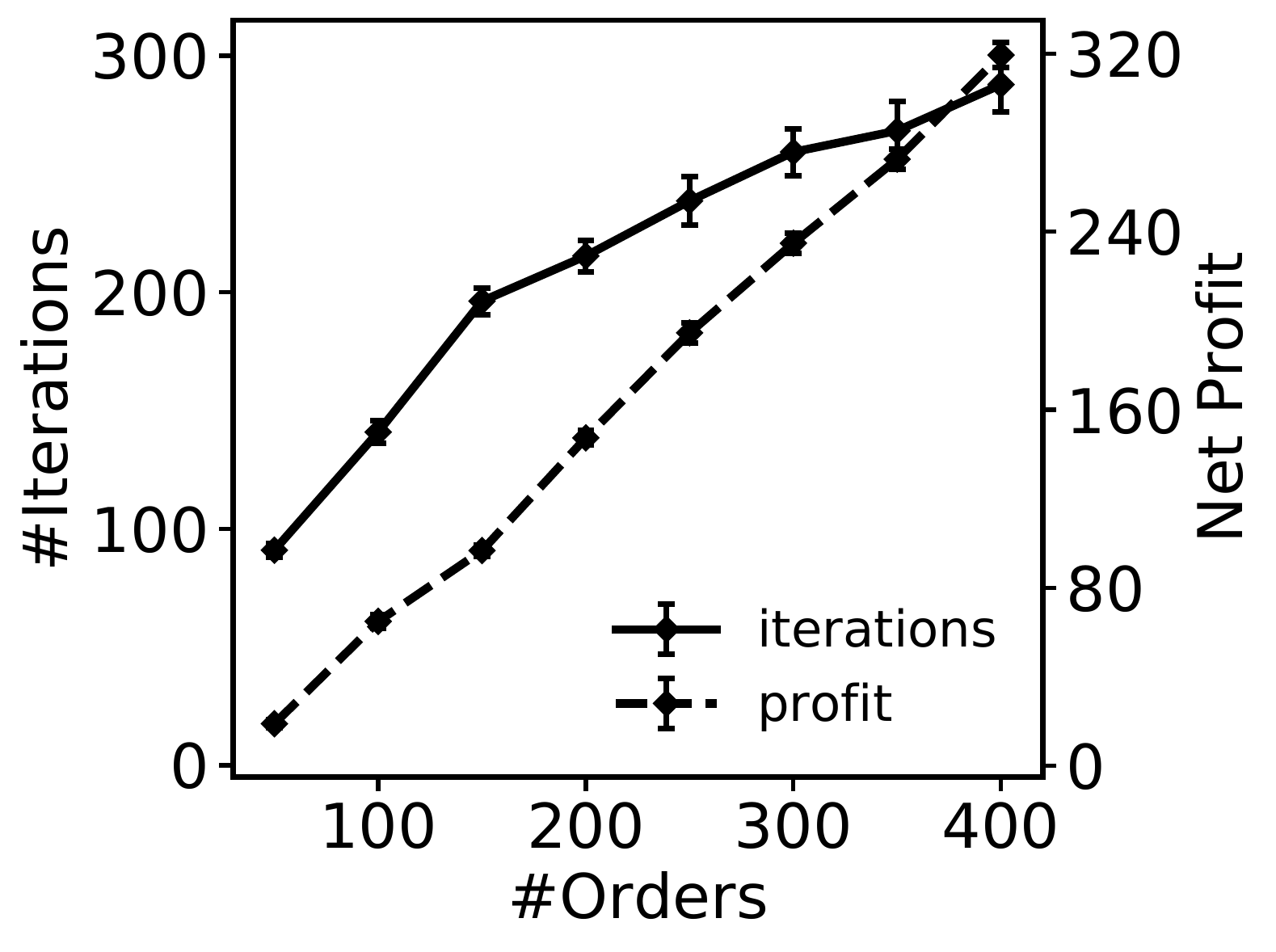}
		\caption{Vary number of orders $n_\text{orders}$ in markets with $U=4$ and $\eta = 2^{-4}$.}
		\label{subfig:combo_plots_b}
		\vspace{2ex}
	\end{subfigure}
	\hspace{0.02\columnwidth}
	\begin{subfigure}{0.48\columnwidth}	
		\centering
		\includegraphics[width=0.9\columnwidth]{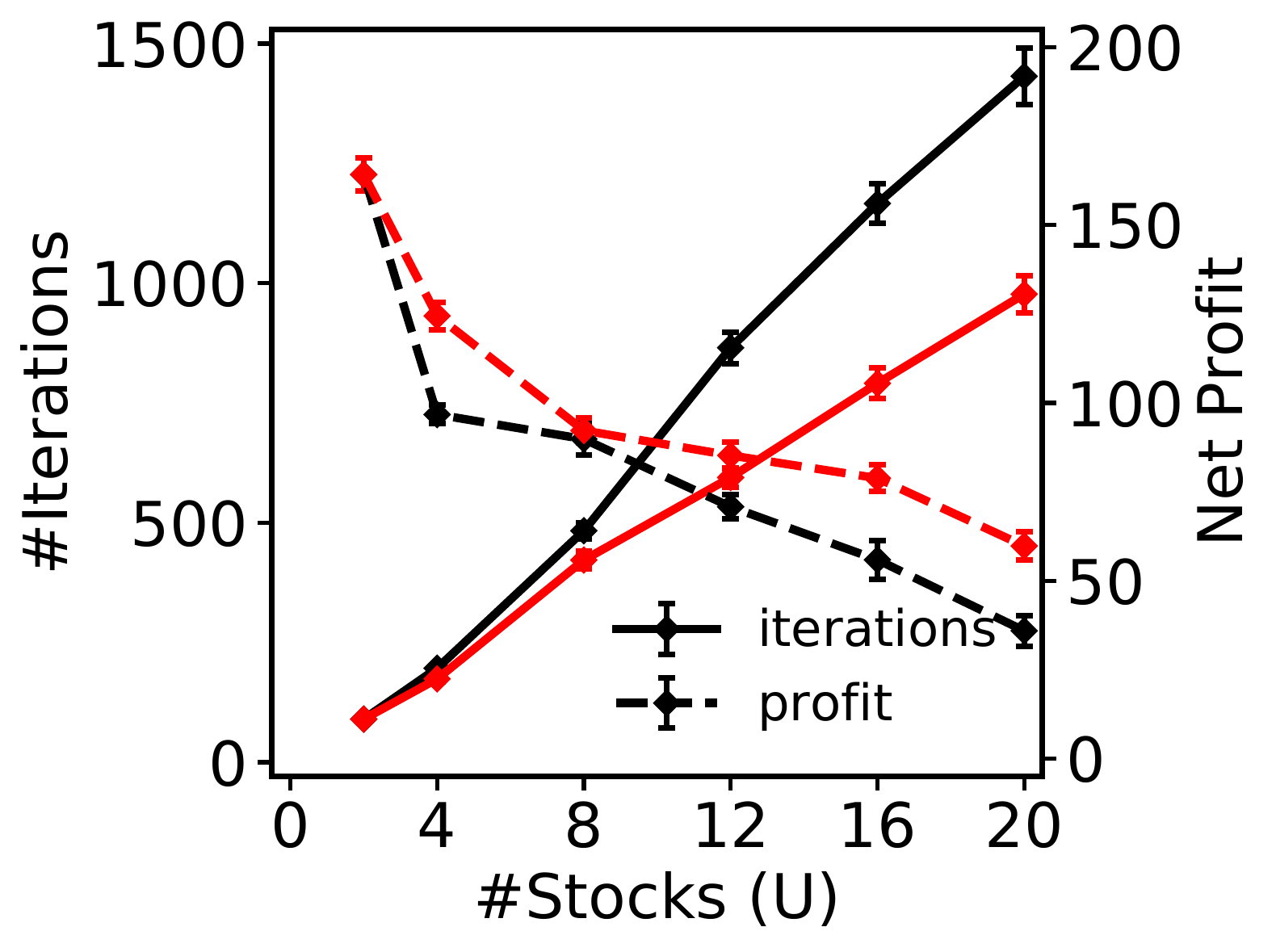}
		\caption{Vary size of underlying assets $U$ in markets with $n_\text{orders} = 150$ and $\eta = 2^{-4}$.}
		\label{subfig:combo_plots_c}
	\end{subfigure}
	\caption[Results of using \Mech{combo_match} to match orders in synthetic combinatorial options markets.]{Results of using \Mech{combo_match} to match orders in synthetic combinatorial options markets. The number of generated constraints (solid lines) and the net profits (dashed line), as the markets vary in price noise, the number of orders, and the size of underlying assets. Red lines represent markets that offer a restrictive set of asset pairs, which covers all $U$ underlying assets.}
	\label{fig:combo_plots}
\end{figure}
We explore a wide range of markets that vary in price noise $\eta$, the number of orders $n_\text{orders}$, and the number of underlying assets $U$.
For each market, we measure (1)~the number of iterations (i.e., the total number of constraints generated) that \Algo{comb_match} takes to find an exact optimal clearing and (2)~the net profit made from the trade.
For all experiments, we show results averaged over 40 simulated markets, with the error bars denoted one standard error around the means.

We first validate that as larger noise is added to the derived bids and asks, the likelihood of matching in our simulated combinatorial options market becomes higher.
We generate markets with four underlying assets (arbitrarily selected from the 30 stocks in DJI) and 150 synthetic combinatorial options orders, and vary the noise level $\eta \in \{2^{-7}, 2^{-6}, 2^{-5}, 2^{-4}, 2^{-3}\}$. 
Figure~\ref{subfig:combo_plots_a} plots the averaged results. 
As expected, the net profit made from the optimal match increases, as $\eta$ increases.
Moreover, we find that as the matching probability increases, the number of iterations that takes to find the optimal solution consistently decreases.
This makes sense as intuitively, in thin markets where few trades are likely to occur, the lower-level optimization program will keep coming up with candidate $\S$ values to refute a large number of matching proposals until convergence.

Figure~\ref{subfig:combo_plots_b} further quantifies the change in iteration numbers and net profits, as we vary the number of combinatorial options orders.
We fix these markets to have four underlying assets with a price noise of $2^{-4}$. 
As we see from Figure~\ref{subfig:combo_plots_b}, as a market aggregates more orders, transactions are more likely to happen, leading to larger net profits.
We also find that the number of generated constraints grows (sub)linear in the number of orders.
Since different $\S$ values are generated to define payoffs of \emph{distinct} options and the number of distinct options increases sublinear in the number of total orders, the rate of increase in the number of iterations tends to decrease as a market aggregates more orders.

Finally, we evaluate how \Algo{comb_match} scales to markets with increasingly larger numbers of underlying assets. 
In this case, as the dimension of $\S$ becomes large, the number of asset-value combinations grows exponentially.
Figure~\ref{subfig:combo_plots_c} (black lines) demonstrates a much faster increase in the number of iterations and a steady decrease in the net profit.%
\footnote{We report average runtimes to quantify the impact of increasing constraints.
	The average times (in seconds) that \Algo{comb_match} computes the optimal match are 4, 31, 47, 59, 68, and 72 for the respective markets with $U \in \{2, 4, 8, 12, 16, 20\}$. 
}
It suggests that as the market provides a large set of underlying assets (e.g., all 30 stocks in DJI), the thin market problem may still arise even when the mechanism facilitates matching options written on different combinations of underlying assets.
Here, in this set of experiments, we make the assumption that every asset pair in the $U \choose 2$ is equally likely to be traded.
In real markets, investors may be more interested in certain asset pairs, trading them more frequently than the others.
Based on such observations, a market can specify a prescriptive set of asset pairs, $\cP$, to alleviate the thin market problem.
For the experiments, we choose $\card{\cP} = U$ and have underlying securities in those asset pairs cover the $U$ assets.
Traders can choose from any asset pairs within this prescriptive set and specify custom weights for each security.
Figure~\ref{subfig:combo_plots_c} (red lines) shows that such prescriptive design may indeed attenuate the thin market problem, leading to a faster convergence and higher profits.
\section{Discussion}
\label{sec:conclusion}
The paper proposes market designs to facilitate trading standard options and the more general \emph{combinatorial financial options}.
We start by examining standard financial exchanges that operate separate options markets and have each independently aggregate and match orders on options of a designated type, strike price, and expiration.
When logically related financial markets run independently, traders may remove arbitrage and close bid-ask spreads themselves. 
Our OptionMetrics experiments show that they do so suboptimally.
Execution risk and transaction cost often prevent arbitrage opportunities from being fully removed.
Moreover, profits 
can flow to agents with better computational power and little information. 
Our consolidated market design, by putting computational power into the exchange, supports trading related options across different strikes and reduces bid-ask spreads, thus reducing the arms race among agents and rewarding informed traders only.

We extend standard options to define combinatorial financial options.
We are interested in designing a fully expressive combinatorial options market that allows options to be specified on all linear combinations of assets. 
In such a market, traders are granted the convenience and expressiveness to speculate on correlated movements among stocks and more precisely hedge their risks (e.g., in one transaction, purchase a put option on their exact investment portfolio).
Increasing expressiveness may also benefit the exchange: the market is able to incorporate information of greater detail to optimize outcome, improving economic efficiency, and obtain a surplus from trade to split between the exchange and traders.
Like many other combinatorial markets, the combinatorial options market comes at the cost of a higher computational complexity: optimal clearing of such a market is coNP-hard.
We propose a constraint generation algorithm to solve for the exact optimal match, and demonstrate its viability on synthetically generated combinatorial options markets of different scales. 

Two immediate questions arise from our work.
\emph{First}, can we design computationally efficient matching algorithms by limiting expressivity within a combinatorial options market or imposing certain assumptions on the underlying assets?
One next step is to explore naturally structured markets where combinations are limited to components in a graph of underlying assets. 
One special case is a hierarchical graph \cite{Guo2009}, for example, the \snp, sectors like travel and technology, subsectors like airlines and internet within those, etc. 
Another direction is to relax our current constraint that guarantees no loss for all possible states, assuming or learning a probability distribution over the security's future price (e.g., recovering probability from option prices~\cite{recover_prob_options}).

\emph{Second}, when operating such a combinatorial options market, how should an exchange set the market clearing rule?
While the proposed mechanism (\ref{mech:combo_match}) and matching algorithm (\Algo{comb_match}) allow flexibility in the timing of market clear operations, and so can work in either continuous or periodic (batch) form, adopting an appropriate clearing rule is important and can affect both computational and economic efficiency.
Our experiments on synthetic combinatorial options market (Section~\ref{sec:exp_combo_option}) provide some preliminary understanding.
Results on markets with low noise suggest potential limitations of continuous clearing: as there is no match among existing orders and the matching probability is low with only one incoming order, the lower-level program may take a long time to refute a large number of matching proposals until convergence. 
Surplus can also be low in a continuous clearing market that looks for immediate trades.%
\footnote{We conduct preliminary experiments on synthetic data and find that the average net profit achieved in a market with continuous clearning is about 30\% lower than that of a market with batch clearing of the same 100 orders.}
Batch clearing might be suitable for this combinatorial market, achieving a more practical convergence speed and higher surplus.
The next question becomes how frequently should the market clear? 
Like discussed in prior works studying the proposal of having financial markets move from continuous to batch clearing \cite{budish2014,budish2015,brinkman2017empirical}, choosing the appropriate clearing interval is an art.
In our case, it depends on the market scale (e.g., the number of underlying assets and orders) and liquidity (e.g., matching probability), and a designer needs to balance price discovery and matching surplus (longer batches improve trading surplus, but can significantly slow down price discovery).
Analysis based on empirical mechanism design~\cite{emd,brinkman2017empirical} may be useful to help a market designer set the optimal clearing interval and maximize efficiency.

We leave these questions open for future research.

\begin{acks}
We thank Miroslav Dud\'ik and Haoming Shen for helpful discussions. 
We are also grateful to the anonymous reviewers for constructive feedback. 
\end{acks}

\bibliographystyle{ACM-Reference-Format}
\bibliography{refs}


\begin{thebibliography}{31}


\ifx \showCODEN    \undefined \def \showCODEN     #1{\unskip}     \fi
\ifx \showDOI      \undefined \def \showDOI       #1{#1}\fi
\ifx \showISBNx    \undefined \def \showISBNx     #1{\unskip}     \fi
\ifx \showISBNxiii \undefined \def \showISBNxiii  #1{\unskip}     \fi
\ifx \showISSN     \undefined \def \showISSN      #1{\unskip}     \fi
\ifx \showLCCN     \undefined \def \showLCCN      #1{\unskip}     \fi
\ifx \shownote     \undefined \def \shownote      #1{#1}          \fi
\ifx \showarticletitle \undefined \def \showarticletitle #1{#1}   \fi
\ifx \showURL      \undefined \def \showURL       {\relax}        \fi
\providecommand\bibfield[2]{#2}
\providecommand\bibinfo[2]{#2}
\providecommand\natexlab[1]{#1}
\providecommand\showeprint[2][]{arXiv:#2}

\bibitem[\protect\citeauthoryear{Benisch, Sadeh, and Sandholm}{Benisch
  et~al\mbox{.}}{2008}]%
        {Benisch2008}
\bibfield{author}{\bibinfo{person}{Michael Benisch}, \bibinfo{person}{Norman
  Sadeh}, {and} \bibinfo{person}{Tuomas Sandholm}.}
  \bibinfo{year}{2008}\natexlab{}.
\newblock \showarticletitle{A theory of expressiveness in mechanisms}. In
  \bibinfo{booktitle}{\emph{Proceedings of the 23rd National Conference on
  Artificial Intelligence}}. \bibinfo{pages}{17--23}.
\newblock


\bibitem[\protect\citeauthoryear{Black and Scholes}{Black and Scholes}{1973}]%
        {bs_model}
\bibfield{author}{\bibinfo{person}{Fischer Black} {and} \bibinfo{person}{Myron
  Scholes}.} \bibinfo{year}{1973}\natexlab{}.
\newblock \showarticletitle{The pricing of options and corporate liabilities}.
\newblock \bibinfo{journal}{\emph{Journal of Political Economy}}
  \bibinfo{volume}{81}, \bibinfo{number}{3} (\bibinfo{year}{1973}),
  \bibinfo{pages}{637--654}.
\newblock


\bibitem[\protect\citeauthoryear{Brinkman and Wellman}{Brinkman and
  Wellman}{2017}]%
        {brinkman2017empirical}
\bibfield{author}{\bibinfo{person}{Erik Brinkman} {and}
  \bibinfo{person}{Michael~P Wellman}.} \bibinfo{year}{2017}\natexlab{}.
\newblock \showarticletitle{Empirical mechanism design for optimizing clearing
  interval in frequent call markets}. In \bibinfo{booktitle}{\emph{Proceedings
  of the 2017 ACM Conference on Economics and Computation}}.
  \bibinfo{pages}{205--221}.
\newblock


\bibitem[\protect\citeauthoryear{Budish, Cramton, and Shim}{Budish
  et~al\mbox{.}}{2014}]%
        {budish2014}
\bibfield{author}{\bibinfo{person}{Eric Budish}, \bibinfo{person}{Peter
  Cramton}, {and} \bibinfo{person}{John Shim}.}
  \bibinfo{year}{2014}\natexlab{}.
\newblock \showarticletitle{Implementation details for frequent batch auctions:
  Slowing down markets to the blink of an eye}.
\newblock \bibinfo{journal}{\emph{American Economic Review}}
  \bibinfo{volume}{104}, \bibinfo{number}{5} (\bibinfo{date}{May}
  \bibinfo{year}{2014}), \bibinfo{pages}{418--424}.
\newblock


\bibitem[\protect\citeauthoryear{Budish, Cramton, and Shim}{Budish
  et~al\mbox{.}}{2015}]%
        {budish2015}
\bibfield{author}{\bibinfo{person}{Eric Budish}, \bibinfo{person}{Peter
  Cramton}, {and} \bibinfo{person}{John Shim}.}
  \bibinfo{year}{2015}\natexlab{}.
\newblock \showarticletitle{{The high-frequency trading arms race: frequent
  batch auctions as a market design response}}.
\newblock \bibinfo{journal}{\emph{The Quarterly Journal of Economics}}
  \bibinfo{volume}{130}, \bibinfo{number}{4} (\bibinfo{date}{07}
  \bibinfo{year}{2015}), \bibinfo{pages}{1547--1621}.
\newblock


\bibitem[\protect\citeauthoryear{Cao and Wei}{Cao and Wei}{2010}]%
        {Cao2010}
\bibfield{author}{\bibinfo{person}{Melanie Cao} {and} \bibinfo{person}{Jason
  Wei}.} \bibinfo{year}{2010}\natexlab{}.
\newblock \showarticletitle{Option market liquidity: Commonality and other
  characteristics}.
\newblock \bibinfo{journal}{\emph{Journal of Financial Markets}}
  \bibinfo{volume}{13}, \bibinfo{number}{1} (\bibinfo{year}{2010}),
  \bibinfo{pages}{20--48}.
\newblock


\bibitem[\protect\citeauthoryear{Chen, Fortnow, Lambert, Pennock, and
  Wortman}{Chen et~al\mbox{.}}{2008a}]%
        {Chen2008b}
\bibfield{author}{\bibinfo{person}{Yiling Chen}, \bibinfo{person}{Lance
  Fortnow}, \bibinfo{person}{Nicolas Lambert}, \bibinfo{person}{David~M.
  Pennock}, {and} \bibinfo{person}{Jennifer Wortman}.}
  \bibinfo{year}{2008}\natexlab{a}.
\newblock \showarticletitle{Complexity of combinatorial market makers}. In
  \bibinfo{booktitle}{\emph{Proceedings of the 9th ACM Conference on Electronic
  Commerce}}. \bibinfo{pages}{190--199}.
\newblock


\bibitem[\protect\citeauthoryear{Chen, Fortnow, Nikolova, and Pennock}{Chen
  et~al\mbox{.}}{2007}]%
        {Chen2007}
\bibfield{author}{\bibinfo{person}{Yiling Chen}, \bibinfo{person}{Lance
  Fortnow}, \bibinfo{person}{Evdokia Nikolova}, {and} \bibinfo{person}{David~M.
  Pennock}.} \bibinfo{year}{2007}\natexlab{}.
\newblock \showarticletitle{Betting on permutations}. In
  \bibinfo{booktitle}{\emph{Proceedings of the 8th ACM conference on Electronic
  Commerce}}. \bibinfo{pages}{326--335}.
\newblock


\bibitem[\protect\citeauthoryear{Chen, Goel, and Pennock}{Chen
  et~al\mbox{.}}{2008b}]%
        {Chen2008a}
\bibfield{author}{\bibinfo{person}{Yiling Chen}, \bibinfo{person}{Sharad Goel},
  {and} \bibinfo{person}{David~M. Pennock}.} \bibinfo{year}{2008}\natexlab{b}.
\newblock \showarticletitle{Pricing combinatorial markets for tournaments}. In
  \bibinfo{booktitle}{\emph{Proceedings of the 40th Annual ACM Symposium on
  Theory of Computing}}. \bibinfo{pages}{305--314}.
\newblock


\bibitem[\protect\citeauthoryear{Dud\'ik, Lahaie, Pennock, and
  Rothschild}{Dud\'ik et~al\mbox{.}}{2013}]%
        {DudikEtAl13}
\bibfield{author}{\bibinfo{person}{Miroslav Dud\'ik},
  \bibinfo{person}{S\'ebastien Lahaie}, \bibinfo{person}{David~M. Pennock},
  {and} \bibinfo{person}{David Rothschild}.} \bibinfo{year}{2013}\natexlab{}.
\newblock \showarticletitle{A combinatorial prediction market for the {U}.{S}.
  elections}. In \bibinfo{booktitle}{\emph{Proceedings of the 14th ACM
  Conference on Electronic Commerce}}. \bibinfo{pages}{341--358}.
\newblock


\bibitem[\protect\citeauthoryear{Dud\'ik, Wang, Pennock, and
  Rothschild}{Dud\'ik et~al\mbox{.}}{2020}]%
        {Dudik2020}
\bibfield{author}{\bibinfo{person}{Miroslav Dud\'ik}, \bibinfo{person}{Xintong
  Wang}, \bibinfo{person}{David~M. Pennock}, {and} \bibinfo{person}{David~M.
  Rothschild}.} \bibinfo{year}{2020}\natexlab{}.
\newblock \showarticletitle{Log-time prediction markets for interval
  securities}. In \bibinfo{booktitle}{\emph{Proceedings of the 20th
  International Conference on Autonomous Agents and Multiagent Systems}}.
  \bibinfo{pages}{465--473}.
\newblock


\bibitem[\protect\citeauthoryear{Fortnow, Kilian, Pennock, and Wellman}{Fortnow
  et~al\mbox{.}}{2005}]%
        {Fortnow2005}
\bibfield{author}{\bibinfo{person}{Lance Fortnow}, \bibinfo{person}{Joe
  Kilian}, \bibinfo{person}{David~M. Pennock}, {and}
  \bibinfo{person}{Michael~P. Wellman}.} \bibinfo{year}{2005}\natexlab{}.
\newblock \showarticletitle{Betting Boolean-style: A framework for trading in
  securities based on logical formulas}.
\newblock \bibinfo{journal}{\emph{Decision Support Systems}}
  \bibinfo{volume}{39}, \bibinfo{number}{1} (\bibinfo{year}{2005}),
  \bibinfo{pages}{87--104}.
\newblock


\bibitem[\protect\citeauthoryear{Garman}{Garman}{1976}]%
        {Garman1976}
\bibfield{author}{\bibinfo{person}{Mark~B. Garman}.}
  \bibinfo{year}{1976}\natexlab{}.
\newblock \showarticletitle{An algebra for evaluating hedge portfolios}.
\newblock \bibinfo{journal}{\emph{Journal of Financial Economics}}
  \bibinfo{volume}{3}, \bibinfo{number}{4} (\bibinfo{year}{1976}),
  \bibinfo{pages}{403--427}.
\newblock


\bibitem[\protect\citeauthoryear{Golovin}{Golovin}{2007}]%
        {Golovin2007}
\bibfield{author}{\bibinfo{person}{Daniel Golovin}.}
  \bibinfo{year}{2007}\natexlab{}.
\newblock \showarticletitle{More expressive market models and the future of
  combinatorial auctions}.
\newblock \bibinfo{journal}{\emph{SIGecom Exchanges}} \bibinfo{volume}{7},
  \bibinfo{number}{1} (\bibinfo{year}{2007}), \bibinfo{pages}{55--57}.
\newblock


\bibitem[\protect\citeauthoryear{Guo and Pennock}{Guo and Pennock}{2009}]%
        {Guo2009}
\bibfield{author}{\bibinfo{person}{Mingyu Guo} {and} \bibinfo{person}{David~M.
  Pennock}.} \bibinfo{year}{2009}\natexlab{}.
\newblock \showarticletitle{Combinatorial prediction markets for event
  hierarchies}. In \bibinfo{booktitle}{\emph{Proceedings of the 8th
  International Conference on Autonomous Agents and Multi-agent Systems}}.
  \bibinfo{pages}{201--208}.
\newblock


\bibitem[\protect\citeauthoryear{Hanson}{Hanson}{2003}]%
        {Hanson03}
\bibfield{author}{\bibinfo{person}{Robin Hanson}.}
  \bibinfo{year}{2003}\natexlab{}.
\newblock \showarticletitle{Combinatorial information market design}.
\newblock \bibinfo{journal}{\emph{Information Systems Frontiers}}
  \bibinfo{volume}{5}, \bibinfo{number}{1} (\bibinfo{year}{2003}),
  \bibinfo{pages}{107--119}.
\newblock


\bibitem[\protect\citeauthoryear{Herzel}{Herzel}{2005}]%
        {Herzel2005}
\bibfield{author}{\bibinfo{person}{Stefano Herzel}.}
  \bibinfo{year}{2005}\natexlab{}.
\newblock \showarticletitle{Arbitrage opportunities on derivatives: A linear
  programming approach}.
\newblock \bibinfo{journal}{\emph{Dynamics of Continuous, Discrete and
  Impulsive Systems Series B: Application and Algorithms}}
  (\bibinfo{year}{2005}).
\newblock


\bibitem[\protect\citeauthoryear{Jackwerth and Rubinstein}{Jackwerth and
  Rubinstein}{1996}]%
        {recover_prob_options}
\bibfield{author}{\bibinfo{person}{Jens~Carsten Jackwerth} {and}
  \bibinfo{person}{Mark Rubinstein}.} \bibinfo{year}{1996}\natexlab{}.
\newblock \showarticletitle{Recovering probability distributions from option
  prices}.
\newblock \bibinfo{journal}{\emph{The Journal of Finance}}
  \bibinfo{volume}{51}, \bibinfo{number}{5} (\bibinfo{year}{1996}),
  \bibinfo{pages}{1611--1631}.
\newblock


\bibitem[\protect\citeauthoryear{Kroer, Dud\'ik, Lahaie, and
  Balakrishnan}{Kroer et~al\mbox{.}}{2016}]%
        {KroerDudik16}
\bibfield{author}{\bibinfo{person}{Christian Kroer}, \bibinfo{person}{Miroslav
  Dud\'ik}, \bibinfo{person}{S\'ebastien Lahaie}, {and}
  \bibinfo{person}{Sivaraman Balakrishnan}.} \bibinfo{year}{2016}\natexlab{}.
\newblock \showarticletitle{Arbitrage-free combinatorial market making via
  integer programming}. In \bibinfo{booktitle}{\emph{Proceedings of the 17th
  ACM Conference on Electronic Commerce}}. \bibinfo{pages}{161--178}.
\newblock


\bibitem[\protect\citeauthoryear{Laskey, Sun, Hanson, Twardy, Matsumoto, and
  Goldfedder}{Laskey et~al\mbox{.}}{2018}]%
        {LaskeyEtAl18}
\bibfield{author}{\bibinfo{person}{Kathryn~Blackmond Laskey},
  \bibinfo{person}{Wei Sun}, \bibinfo{person}{Robin Hanson},
  \bibinfo{person}{Charles Twardy}, \bibinfo{person}{Shou Matsumoto}, {and}
  \bibinfo{person}{Brandon Goldfedder}.} \bibinfo{year}{2018}\natexlab{}.
\newblock \showarticletitle{Graphical model market maker for combinatorial
  prediction markets}.
\newblock \bibinfo{journal}{\emph{Journal of Artificial Intelligence Research}}
   \bibinfo{volume}{63} (\bibinfo{year}{2018}), \bibinfo{pages}{421--460}.
\newblock


\bibitem[\protect\citeauthoryear{Merton}{Merton}{1973}]%
        {Merton1973}
\bibfield{author}{\bibinfo{person}{Robert~C. Merton}.}
  \bibinfo{year}{1973}\natexlab{}.
\newblock \showarticletitle{Theory of rational option pricing}.
\newblock \bibinfo{journal}{\emph{The Bell Journal of Economics and Management
  Science}} \bibinfo{volume}{4}, \bibinfo{number}{1} (\bibinfo{year}{1973}),
  \bibinfo{pages}{141--183}.
\newblock


\bibitem[\protect\citeauthoryear{Modigliani and Miller}{Modigliani and
  Miller}{1958}]%
        {Modigliani1958}
\bibfield{author}{\bibinfo{person}{Franco Modigliani} {and}
  \bibinfo{person}{Merton~H. Miller}.} \bibinfo{year}{1958}\natexlab{}.
\newblock \showarticletitle{The cost of capital, corporation finance and the
  theory of investment}.
\newblock \bibinfo{journal}{\emph{The American Economic Review}}
  \bibinfo{volume}{48}, \bibinfo{number}{3} (\bibinfo{year}{1958}),
  \bibinfo{pages}{261--297}.
\newblock


\bibitem[\protect\citeauthoryear{Optimization}{Optimization}{2020}]%
        {gurobi}
\bibfield{author}{\bibinfo{person}{Gurobi Optimization}.}
  \bibinfo{year}{2020}\natexlab{}.
\newblock \showarticletitle{Gurobi Optimizer Reference Manual}.
\newblock \bibinfo{journal}{\emph{http://www.gurobi.com}}
  (\bibinfo{year}{2020}).
\newblock


\bibitem[\protect\citeauthoryear{Perrakis and Ryan}{Perrakis and Ryan}{1984}]%
        {pricing_discrete_time}
\bibfield{author}{\bibinfo{person}{Stylianos Perrakis} {and}
  \bibinfo{person}{Peter~J. Ryan}.} \bibinfo{year}{1984}\natexlab{}.
\newblock \showarticletitle{Option pricing bounds in discrete time}.
\newblock \bibinfo{journal}{\emph{The Journal of Finance}}
  \bibinfo{volume}{39}, \bibinfo{number}{2} (\bibinfo{year}{1984}),
  \bibinfo{pages}{519--525}.
\newblock


\bibitem[\protect\citeauthoryear{Ritchken}{Ritchken}{1985}]%
        {Ritchken1985}
\bibfield{author}{\bibinfo{person}{Peter~H. Ritchken}.}
  \bibinfo{year}{1985}\natexlab{}.
\newblock \showarticletitle{On option pricing bounds}.
\newblock \bibinfo{journal}{\emph{The Journal of Finance}}
  \bibinfo{volume}{40}, \bibinfo{number}{4} (\bibinfo{year}{1985}),
  \bibinfo{pages}{1219--1233}.
\newblock


\bibitem[\protect\citeauthoryear{Rothschild and Pennock}{Rothschild and
  Pennock}{2014}]%
        {Rothschild14}
\bibfield{author}{\bibinfo{person}{David~M. Rothschild} {and}
  \bibinfo{person}{David~M. Pennock}.} \bibinfo{year}{2014}\natexlab{}.
\newblock \showarticletitle{The extent of price misalignment in prediction
  markets}.
\newblock \bibinfo{journal}{\emph{Algorithmic Finance}} \bibinfo{volume}{3},
  \bibinfo{number}{2} (\bibinfo{year}{2014}), \bibinfo{pages}{3--20}.
\newblock


\bibitem[\protect\citeauthoryear{Sandholm}{Sandholm}{2007}]%
        {Sandholm2007}
\bibfield{author}{\bibinfo{person}{Tuomas Sandholm}.}
  \bibinfo{year}{2007}\natexlab{}.
\newblock \showarticletitle{Expressive commerce and its application to
  sourcing: How we conducted \$35 billion of generalized combinatorial
  auctions}.
\newblock \bibinfo{journal}{\emph{AI Magazine}} \bibinfo{volume}{28},
  \bibinfo{number}{3} (\bibinfo{year}{2007}).
\newblock


\bibitem[\protect\citeauthoryear{Singh}{Singh}{2019}]%
        {Singh2019}
\bibfield{author}{\bibinfo{person}{Charan Singh}.}
  \bibinfo{year}{2019}\natexlab{}.
\newblock \showarticletitle{How often do options get exercised early?}
\newblock
  \bibinfo{journal}{\emph{https://www.optionsanimal.com/how-often-do-options-get-exercised-early}}
  (\bibinfo{year}{2019}).
\newblock


\bibitem[\protect\citeauthoryear{Varian}{Varian}{1987}]%
        {Varian1987}
\bibfield{author}{\bibinfo{person}{Hal~R. Varian}.}
  \bibinfo{year}{1987}\natexlab{}.
\newblock \showarticletitle{The arbitrage principle in financial economics}.
\newblock \bibinfo{journal}{\emph{The Journal of Economic Perspectives}}
  \bibinfo{volume}{1}, \bibinfo{number}{2} (\bibinfo{year}{1987}),
  \bibinfo{pages}{55--72}.
\newblock


\bibitem[\protect\citeauthoryear{Vorobeychik, Kiekintveld, and
  Wellman}{Vorobeychik et~al\mbox{.}}{2006}]%
        {emd}
\bibfield{author}{\bibinfo{person}{Yevgeniy Vorobeychik},
  \bibinfo{person}{Christopher Kiekintveld}, {and} \bibinfo{person}{Michael~P.
  Wellman}.} \bibinfo{year}{2006}\natexlab{}.
\newblock \showarticletitle{Empirical mechanism design: Methods, with
  application to a supply-chain scenario}. In \bibinfo{booktitle}{\emph{7th ACM
  Conference on Electronic Commerce}}. \bibinfo{pages}{306--315}.
\newblock


\bibitem[\protect\citeauthoryear{Xia and Pennock}{Xia and Pennock}{2011}]%
        {XiaPe11}
\bibfield{author}{\bibinfo{person}{Lirong Xia} {and} \bibinfo{person}{David~M.
  Pennock}.} \bibinfo{year}{2011}\natexlab{}.
\newblock \showarticletitle{An efficient Monte-Carlo algorithm for pricing
  combinatorial prediction markets for tournaments}. In
  \bibinfo{booktitle}{\emph{Proceedings of the 22nd International Joint
  Conference on Artificial Intelligence}}. \bibinfo{pages}{452--457}.
\newblock


\end{thebibliography}

\newpage
\appendix
     \section{Deferred Proofs from Section~\ref{sec:standard_option}}     
\subsection{Proof of \Thm{consolidate_standard_options}}
As the left-hand side of the constraint in \ref{mech:single_match} is a piecewise linear function of $S$, it suffices to solve \ref{mech:single_match} by including constraints defined by $S$ at breakpoint values.
In our case, breakpoints of the constraint are the strike values specified in the market, and the number of distinct strikes $n_{K}$ grows sublinear in the number of orders $n_\text{orders}$.
Therefore, we can specify the constraint in \ref{mech:single_match} as $n_K+2$ payoff constraints for each stock value $S \in {\vp} \cup {\vq} \cup \{0, \infty\}$.
Moreover, there are at most $n_\text{orders}$ quantity constraints, if we consider the quantity specified in each order.
Thus, \ref{mech:single_match} is a linear program with $n_{\text{orders}}+1$ decision variables and $\order(n_{\text{orders}})$ constraints.
\Mech{single_match} matches options written on the same underlying asset and expiration across all types and strikes in time polynomial in the number of orders.

\subsection{Proof of Price Quote Procedure}
\label{app:proof_price_quote}
We reiterate the procedure of using mechanism \ref{mech:single_match} to price a target options $(\chi, S, K, T)$ with existing options orders in the market defined by $(\vphi, \vp, \vb, \vpsi, \vq, \va)$.
\begin{enumerate}[(1)]
	\item The best bid $b^*$ for an option $(\chi, S, K, T)$ is the maximum gain of selling a portfolio of options that is \emph{weakly dominated} by $(\chi, S, K, T)$ for some constant $L$.
	
	We derive $b^*$ by adding $(\chi, S, K, T)$ to the sell side of the market indexed $N+1$, initializing its price $a_{N+1}$ to 0, and solving for \ref{mech:single_match}. 
	The best bid $b^*$ is the returned objective of \ref{mech:single_match}.
	
	\begin{proof}
		We show that the above procedure finds the best bid, by returning the maximum gain of selling a \textit{weakly dominated} portfolio of options.
		After adding $(\chi, S, K, T)$ to the market, we have the updated \ref{mech:single_match} as the following:
		\begin{align*}
		\max \limits_{\vgamma, \vdelta, L} & \displaystyle \quad \vb^\top \vgamma - \va^\top \vdelta - a_{N+1}\delta_{N+1} - L\\
		\text{s.t.} & \displaystyle \quad \sum_{m} \gamma_m \max\{\phi_m(S - p_m), 0\} - \sum_{n} \delta_n \max\{\psi_n(S - q_n), 0\} \\
		& \displaystyle \quad - \delta_{N+1} \max\{\chi(S-K), 0\} \leq L \qquad \qquad \qquad \forall S \in [0, \infty)
		\end{align*}
		Since we set $a_{N+1} = 0$, it is always optimal to buy option $(\chi, S, K, T)$ and have $\delta_{N+1} = 1$.
		Therefore, we have the following optimization problem extended from \ref{mech:single_match}:
		\begin{align*}
		\max \limits_{\vgamma, \vdelta, L} & \displaystyle \quad z := \vb^\top \vgamma - \va^\top \vdelta - L \tag{M.1-bid} \label{bid}\\
		\text{s.t.} & \displaystyle \quad \,\,\underbrace{\!\!\sum_{m} \gamma_m \max\{\phi_m(S - p_m), 0\} - \sum_{n} \delta_n \max\{\psi_n(S - q_n), 0\}\!\!}_{\text{Portfolio } (*)}\,\,\\ 
		&\qquad \qquad \qquad \qquad  \qquad \qquad \leq \max\{\chi(S-K), 0\} + L \quad \forall S \in [0, \infty)
		\end{align*}
		Following Definition~\ref{def:opt_dominate}, the left-hand side of the above inequality describes the payoff of a portfolio of existing options that is weakly dominated by the target options $(\chi, S, K, T)$ with an offset $L$, and the objective $z$ maximizes the gain of selling such a portfolio.
		One is willing to buy $(\chi, S, K, T)$ at the highest price (i.e., the best bid) $b^* = z$, as one can always sell the weakly dominated Portfolio ($*$) and get $z$ back without losing anything in the future.
	\end{proof}
	
	\item The best ask $a^*$ for an option $(\chi, S, K, T)$ is the minimum cost of buying a portfolio of options that \emph{weakly dominates} $(\chi, S, K, T)$ for some constant $L$.
	
	We derive $a^*$ by adding $(\chi, S, K, T)$ to the buy side of the market indexed $M+1$, initializing its price $b_{M+1}$ to a large number (i.e., $10^6$), and solving for \ref{mech:single_match}. 
	The best ask $a^*$ is then $b_{M+1}$ minus the returned objective.
	
	\begin{proof}
		We show that the above procedure finds the best ask, by returning the minimum cost of buying a \textit{weakly dominant} portfolio of options.
		After adding $(\chi, S, K, T)$ to the market, we have the updated \ref{mech:single_match} as the following:
		\begin{align*}
		\max \limits_{\vgamma, \vdelta, L} & \displaystyle \quad \vb^\top \vgamma + b_{M+1}\gamma_{M+1} - \va^\top \vdelta - L\\
		\text{s.t.} & \displaystyle \quad \sum_{m} \gamma_m \max\{\phi_m(S - p_m), 0\} + \gamma_{M+1} \max\{\chi(S-K), 0\}\\
		& \displaystyle \quad - \sum_{n} \delta_n \max\{\psi_n(S - q_n), 0\} \leq L \qquad \qquad \forall S \in [0, \infty)
		\end{align*}
		We set $b_{M+1}$ to a sufficiently large number, say $b_{M+1} = 10^6$, so that it is always optimal to sell option $(\chi, S, K, T)$ and thus have $\gamma_{M+1} = 1$.
		Therefore, we have the following optimization problem:
		\begin{align*}
		\max \limits_{\vgamma, \vdelta, L} & \displaystyle \quad z := \vb^\top \vgamma + b_{M+1} - \va^\top \vdelta - L \tag{M.1-ask} \label{ask}\\
		\text{s.t.} & \displaystyle \quad \max\{\chi(S-K), 0\} \leq\\
		& \displaystyle \quad \,\,\underbrace{\!\!-\sum_{m} \gamma_m \max\{\phi_m(S - p_m), 0\} + \sum_{n} \delta_n \max\{\psi_n(S - q_n), 0\}\!\!}_{\text{Portfolio } (*)}\,\, + L\\
		& \qquad \qquad \qquad \qquad \qquad \qquad \qquad \qquad \qquad \qquad \qquad \qquad \qquad \qquad \forall S \in [0, \infty)
		\end{align*}
		Following Definition~\ref{def:opt_dominate}, the right-hand side of the above inequality describes the payoff of a portfolio of existing options that weakly dominates the target options $(\chi, S, K, T)$ with an offset $L$.
		Since $b_{M+1}$ is a fixed constant, the objective that maximizes for $z$ is equivalent to minimizing for $-\vb^\top\vgamma+\va^\top\vdelta+L$, which is the net cost of buying Portfolio ($*$) plus $L$.
		One is willing to sell $(\chi, S, K, T)$ at the lowest price (i.e., the best ask) $a^* = b_{M+1} - z$, as one can always pay $a^*$ and buy back a weakly dominant Portfolio ($*$) without losing anything in the future.
	\end{proof}
\end{enumerate}

\subsection{Proof of Corollary~\ref{coro:frontier_set}}
\label{app:proof_corollary}
We start by showing that in order to quote the most competitive prices (i.e., the highest bid and the lowest ask) for any target option $(\chi, S, K, T)$, it suffices to consider options orders in $\F$.
We prove by contradiction and consider the following two cases:
\begin{enumerate}[(1)]
	\item Suppose that there exists a bid order $o \notin \F$ with $\gamma_o > 0$ in the Portfolio ($*$), which is the optimal portfolio constructed to derive the highest bid or lowest ask for $(\chi, S, K, T)$.\\
	Since $o \notin \F$, then by the contrapositive of Definition~\ref{def:opt_frontier}, the bid of $o$ can be improved by a portfolio of other orders, denoted Portfolio $o^*$, which is weakly dominated by $o$ with some constant offset $L^*$.
	We denote $b_o$ the bid price specified in order $o$ and $\Pi_{o^*}$ the revenue of selling portfolio $o^*$.
	Then, we have 
	\[z^* := \Pi_{o^*} - L^* > b_o \quad \text{and} \quad \Psi_{o^*} \leq \Psi_o + L^*.\]
	This means that we can replace $o$ with Portfolio $o^*$ and the $L^*$ without violating the constraints in \ref{bid} and \ref{ask} (since $\gamma_o \Psi_o \geq \gamma_o(\Psi_{o^*}-L^*)$), and improve the objective by $\gamma_o(z^*-b_o) > 0$.
	This contradicts our premises, and thus to derive the most competitive prices for any option $(\chi, S, K, T)$, we have $\gamma_o = 0$ for all $o \notin \F$.
	
	\item Similarly, suppose that there exists an ask order $o \notin \F$ with $\delta_o > 0$ in the Portfolio ($*$), which is the optimal portfolio constructed to derive the highest bid or lowest ask for $(\chi, S, K, T)$.\\
	Since $o \notin \F$, then by the contrapositive of Definition~\ref{def:opt_frontier}, the ask of $o$ can be improved by a portfolio of other orders, denoted Portfolio $o^*$, which weakly dominates $o$ with some constant offset $L^*$.
	We denote $a_o$ the ask price specified in order $o$ and $\Pi_{o^*}$ the cost of buying portfolio $o^*$.
	Then, we have 
	\[z^* := \Pi_{o^*} - L^* < a_o \quad \text{and} \quad \Psi_{o} \leq \Psi_{o^*} + L^*.\]
	This means that we can replace $o$ with Portfolio $o^*$ and the $L^*$ without violating the constraints in \ref{bid} and \ref{ask} (since $\delta_o \Psi_o \leq \delta_o(\Psi_{o^*}+L^*)$), and improve the objective by $\delta_o(a_o-z^*) > 0$.
	This contradicts our premises, and thus to derive the most competitive prices for any option $(\chi, S, K, T)$, we have $\delta_o = 0$ for all $o \notin \F$.
\end{enumerate}
Therefore, to quote the most competitive prices for any target option $(\chi, S, K, T)$, it suffices to consider options orders in $\F$.
Similar proofs hold for deciding the existence of matching: if an order $o \notin \F$ appears in the matched portfolio, we can always improve the objective by substituting $o$ with Portfolio $o^*$.
Thus, to determine the price quotes and the existence of a match, it suffices to consider options orders in $\F$.
Following \Thm{consolidate_standard_options}, our proposed mechanism \ref{mech:single_match} determines price quotes and the existence of a match in time polynomial in the size of the frontier set.

\section{Deferred Proofs from Section~\ref{sec:combo_options}}
\subsection{A Variation of Theorem~\ref{thm:NP-complete}}
\label{app:proof_thm2}
For the NP-hardness, we prove the following stronger statement, which we will later use directly to prove Theorem~\ref{thm:coNP-hard}.

\begin{theorem} [Variation of Theorem~\ref{thm:NP-complete}]
	\label{thm2var}
	Consider all combinatorial options in the market $(\vphi, \valpha, \vp, \vpsi, \vbeta, \vq)$. 
	For any fixed $L$, it is NP-hard to decide
	\begin{itemize}
		\item Yes: $\vgamma = \vdelta = \vone$ violates constraint \eqref{eq:constraint} in \ref{mech:combo_match} for some $\S$. Moreover, there exists a function $\varepsilon:\mathbb{Z}\to\mathbb{R}^+$ such that given the fixed $L$, for any $(\vgamma,\vdelta)$ that satisfies\vspace{1ex}\\ 
		\vspace{1ex}
		$\sum_{m} \gamma_m \max\{\phi_m(\valpham^\top \S - p_m), 0\} - \sum_{n} \delta_n \max\{\psi_n(\vbetan^\top \S - q_n), 0\} \leq L + 0.25 \quad \forall \S \in \R_{\geq 0}^U$,\\ 
		we have $\frac{|\vgamma|}M<1-\varepsilon(M)$,
		\item No: $\vgamma = \vdelta = \vone$ satisfies constraint \eqref{eq:constraint} for all $\S$,
	\end{itemize}
	even assuming that each combinatorial option is written on at most two underlying assets.
\end{theorem}

\begin{proof}
	We follow the construction in Theorem~\ref{thm:NP-complete}.
	To conclude the proof for the \emph{Yes} instance case, we find the function $\varepsilon(\cdot)$ such that for \textit{any} $(\vgamma,\vdelta)$ with $|\vgamma|\geq M(1-\varepsilon(M))$, there exists a $\S$ that violates constraint  \eqref{eq:constraint} even if the $L$ on the right-hand side is changed to $L+0.25$.
	We keep the assignment of $\S$ as before, and the left-hand side of constraint  \eqref{eq:constraint} is as the following:
	\begin{align*}
	z' := & \sum_{i\in V} \left(\gamma_i f_i - \delta_i^{(1)} g_i^{(1)} - \delta_i^{(2)} g_i^{(2)} - \delta_i^{(3)} g_i^{(3)}\right) + \sum_{e \in E}(\gamma_e f_e - \delta_e g_e) - \delta^\star g^\star\\
	\geq & \sum_{i\in V} \left(\gamma_i f_i - g_i^{(1)} - g_i^{(2)} - g_i^{(3)}\right) + \sum_{e\in E}(\gamma_e f_e - g_e) - g^\star \tag{since $\delta_n\in[0,1]$ and option payoffs are non-negative}\\
	= & z - \sum_{i\in V} (1-\gamma_i) f_i - \sum_{e\in E}(1-\gamma_e) f_e\\
	\geq & z - M \cdot 2K_1 + \sum_{i\in V} \gamma_i \cdot 2K_1 + \sum_{e\in E} \gamma_e \cdot 2K_1 \tag{since $\forall m:f_m\leq 2K_1$ and $M=|V|+|E|$ }\\
	= & z - M \cdot 2K_1 + |\vgamma| \cdot 2K_1 \tag{$|\vgamma| = \sum_{m} \gamma_m$}\\
	\geq & z - M \varepsilon(M) \cdot 2K_1 \tag{since $|\vgamma|\geq M(1-\varepsilon(M))$}\\
	=&L+0.5-M\varepsilon(M)\cdot 2K_1
	\end{align*}
	It suffices to choose $\varepsilon$ such that $M\varepsilon(M)\cdot2K_1<0.25$.
	Recall that $K_1=10|E| < 10M$.
	We can choose, say $\varepsilon=\frac1{80M^2}$.
\end{proof}
\subsection{Proof of Theorem~\ref{thm:coNP-hard}}
We reduce this decision problem from the decision problem in Theorem~\ref{thm2var} with $L=0$ and each combinatorial option is written on at most two underlying assets.

The reduction is as follows.
The instance of the optimization problem have the same $\valpha,\vbeta,\vp,\vq$ as given in the instance for the decision problem.
In other words, the reduction keeps the same for $f_1,\ldots,f_M,g_1,\ldots,g_N$.
Set $a_1=\cdots=a_N=a$, and set $b_1=\cdots=b_M=b$, where $a>0$ is sufficiently small and $b>0$ is sufficiently large.
We will decide both values later.

If the decision problem instance is a \emph{No} instance, we know that the constraint~(2) holds for $f_1,\ldots,f_M$, $g_1,\ldots,g_N$,  $\gamma_1=\cdots=\gamma_M=\delta_1=\cdots=\delta_N=1$ and $L=0$.
Under this feasible assignment for $\vgamma,\vdelta$ and $L$, we have $\vb^\top \vgamma - \va^\top \vdelta - L=Mb-Na$.

If the decision problem instance problem is a \emph{Yes} instance, we aim to show that $\vb^\top \vgamma - \va^\top \vdelta - L<Mb-Na$ for any feasible $\vgamma,\vdelta,L$.
We discuss three different cases: $L>0.25$, $0\leq L\leq0.25$, and $L<0$.

For $L>0.25$, we have $\vb^\top \vgamma - \va^\top \vdelta - L<\vb^\top \vgamma-0.25$.
Since each entry of $\vgamma$ is at most $1$, the maximum of $\vb^\top\vgamma$ is $Mb$.
Putting together, we have $\vb^\top \vgamma - \va^\top \vdelta - L<Mb-0.25$, which is less than $Mb-Na$ if $a$ is set such that $a<1/4N$.
We will fix $a=1/8N$ from now on.

For $0\leq L\leq0.25$, Theorem~\ref{thm2var} ensures that there exists an $\varepsilon$ which depends only on $M$ such that any feasible $\vgamma,\vdelta$ satisfy $|\vgamma|<(1-\varepsilon)M<M$.
Notice that $L$ is set to $0$ in the instance we are reducing \emph{from}, and $L$ here is between $0$ and $0.25$.
These make the statement corresponding to the \emph{Yes} case of Theorem~\ref{thm2var} apply.
Therefore, the objective $\vb^\top\vgamma-\va^\top\vdelta-L\leq\vb^\top\vgamma\leq (1-\varepsilon)Mb$ is strictly less than $Mb-Na$ if $b$ is set such that $b>\frac{Na}{\varepsilon M}$ (notice that $\varepsilon$ in Theorem~\ref{thm2var} does not depend on b).

For $L<0$, notice that substituting $\vgamma=\mathbf{0},\vdelta=\mathbf{1},\S=\mathbf{0}$ to the left-hand side of (2) gives an upper-bound to $-L$.
Let $L^\ast$ be this upper-bound. 
Notice that $L^\ast$ only depends on $\valpha,\vbeta,\vp,\vq$ and is computable in polynomial time.
In the case the decision problem instance is a \emph{Yes} instance, we know that any feasible $\vgamma,\vdelta$ satisfy $|\vgamma|<(1-\varepsilon)M$ (this is already the case for $L=0$, and the feasible region for $(\vgamma,\vdelta)$ can only be smaller for negative $L$).
Therefore, we have $\vb^\top\vgamma-\va^\top\vdelta-L<Mb(1-\varepsilon)+L^\ast$, which is less than $Mb-Na$ if $b$ satisfies $b>\frac{Na+L^\ast}{\varepsilon M}$.
By setting $b=\frac{Na+L^\ast}{\varepsilon M}+1$, the theorem concludes.
Notice that $L^\ast$ and $\varepsilon$ only depend on $\valpha,\vbeta,\vp,\vq$, so our definition of $b$ is valid.

\newpage
\section{Deferred Experimental Results}
\label{app:combo_options_exp}
\subsection{Statistics of options on each stock using M.1 with L as a decision variable}

\begin{table}[H]
	\caption{Summary statistics (matching) of options on each stock in DJI by consolidating options related to the same underlying asset and expiration date.}
	\small
	\setlength\tabcolsep{3pt}
	\setlength{\textfloatsep}{3pt}
	\begin{center}
		\begin{tabular}{p{0.08\textwidth} | p{0.12\textwidth} p{0.11\textwidth} p{0.12\textwidth} | p{0.10\textwidth} p{0.10\textwidth} p{0.10\textwidth} p{0.12\textwidth}}
			\hline
			\rule{0pt}{10pt}
			Stock &\#markets &\#expirations &\#markets per expiration & \#matches ($L \geq 0$) & avg profit &\#matches ($L < 0$) & avg interest rate (\%)\\
			\hline\hline
			AAPL  & 1452  & 14   &  104   & 1   & 0.8    & 3     & 0.94\\ 
			AXP  & 804  & 11   &  73   & 4   & 0.5    & 2     & 0.99\\ 
			BA  & 1694  & 14   &  121   & 5   & 3.35    & 3     &0.66\\ 
			CAT  & 968  & 13   &  74   & 6   & 0.18    & 2     &0.69\\ 
			CSCO  &728  & 12   &  61   & 1   & 0.02    & 2     &0.44\\ 
			CVX   &742  & 12   &  62   & 7   & 1.26    & 2     &0.41\\ 
			DD   &778  & 13   &  60   & 1   & 0.32    & 3     &0.45\\ 
			DIS  &806  & 13   &  62   & 5   & 0.86    & 0     &0\\ 
			GS   &1110  & 12   &  93   & 1   & 0.16    & 3     &0.82\\ 
			HD   &908  & 13   &  70   & 1   & 0.02    & 2    &0.17\\ 
			IBM  &848  & 12   &  71   & 4   & 3.45    & 1     &0.91\\ 
			INTC  &662  & 12   &  55   & 2   & 0.54    & 1     &0.86\\ 
			JNJ  &850  & 12   &  71   & 4   & 0.12    & 2     &0.57\\ 
			JPM  &854  & 13   &  66   & 3   & 0.4    & 3     &0.7\\ 
			KO  &618  & 13  &  48   & 4   & 0.2    & 2     &0.34\\ 
			MCD  &692  & 12   & 58   & 1   & 0.8    & 3     &0.44\\ 
			MMM  &804  & 12   &  67   & 5   & 2.23    & 2     &0.41\\ 
			MRK  &766  & 12   &  64   & 3   & 0.18    & 2     &0.42\\ 
			MSFT  &1194  & 13   & 92   & 1   & 0    & 0     &0\\ 
			NKE  &844  & 12   &  70   & 0   & 0    & 2     &0.6\\ 
			PFE  &640  & 12   &  53   & 6   & 0.84    & 1     &0.59\\ 
			PG  &786  & 12   &  66   & 3   & 0.27    & 2     &0.31\\ 
			RTX  &920  & 14   &  66   & 1   & 0.01    & 0     &0\\ 
			TRV  &256  & 6   &  43   & 0   & 0    & 1     &0.4\\ 
			UNH  &964  & 12   &  80   & 1   & 0.04    & 2     &1.34\\ 
			V  &856  & 13   &  66   & 3   & 0.03    & 5     &1.47\\ 
			VZ  &508  & 12   &  42   & 6   & 0.65    & 1    &1.81\\ 
			WBA  &808  & 11   &  73   & 0   & 0    & 1     &0.14\\ 
			WMT  &810  & 12   &  68   & 3   & 0.2    & 3     &0.51\\ 
			XOM  &832  & 12   &  69   & 7   & 2.78    & 1     &0.85\\ 
			\hline
			Total &25502  & 366   &  69   & 94   & 1.03    & 56     &0.70\\ 
			\hline\hline
		\end{tabular}
	\end{center}
\end{table}
\newpage
\begin{table}[H]
	\caption{Summary statistics (quoting) of options on each stock in DJI by consolidating options related to the same underlying asset and expiration date.}
	\small
	\setlength\tabcolsep{3pt}
	\setlength{\textfloatsep}{3pt}
	\begin{center}
		\begin{tabular}{p{0.08\textwidth} | p{0.08\textwidth} p{0.08\textwidth} p{0.08\textwidth} | p{0.08\textwidth} p{0.08\textwidth} p{0.12\textwidth} p{0.12\textwidth} p{0.10\textwidth}}
			\hline
			\rule{0pt}{10pt}
			Stock &\#option series &\#orders in $\F$ & Frontier (\%) & avg call spread & avg put spread & improved call spread & improved put spread & \% spread reduced\\
			\hline\hline
			AAPL  & 1038  & 787   &  38   & 1.46   & 1.92    & 0.25     & 0.29   & 84\\ 
			AXP  & 476  & 440   &  46   & 0.37   & 0.3    & 0.15     & 0.13   & 58\\ 
			BA  & 734  & 684   & 47   & 1.18   & 0.56    & 0.34     &0.26  & 65\\ 
			CAT  & 482  & 503   &  52   & 0.73   & 0.41    & 0.22     &0.18   & 65\\ 
			CSCO  &608  & 616   & 51   & 1.27   & 0.45    & 0.1     &0.07   & 90\\ 
			CVX   &254  & 274   &  54   & 0.2   & 0.2    & 0.11     &0.09    & 51\\ 
			DD   &564  & 624   &  55   & 0.28   & 0.17    & 0.13     &0.1  & 49\\ 
			DIS  &518  & 515   &  50   & 0.91   & 0.79    & 0.13     &0.12  & 85\\ 
			GS   &738  & 631   &  43   & 0.91   & 0.61    & 0.24     &0.16  & 73\\ 
			HD   &688  & 675   &  49   & 1.69   & 1.53    & 0.33    &0.32   & 80\\ 
			IBM  &506  &429   &  42   & 1.54   & 0.76    & 0.31     &0.21    & 77\\ 
			INTC  &462  & 523   &  57   & 1.07   & 0.62    & 0.2     &0.16   & 79\\ 
			JNJ  &404  & 412   &  51   & 1.04   & 1.2    & 0.14     &0.13  & 88\\ 
			JPM  &524  & 472   &  45   & 0.62   & 0.57    & 0.11     &0.09   & 83\\ 
			KO  &370  & 350  &  47   & 0.34   & 0.17    & 0.05     &0.04   & 82\\ 
			MCD  &266  & 285   & 54   & 0.48   & 0.31    & 0.21     &0.16    & 53\\ 
			MMM  &336  & 379   &  56   & 0.74   & 0.56    & 0.22     &0.19   & 68\\ 
			MRK  &488  & 507   &  52   & 0.4   & 0.34    & 0.15    &0.14   & 62\\ 
			MSFT  &1120  & 833   & 37   & 2.2   & 0.96    & 0.33     &0.27   &  81\\ 
			NKE  &720  & 766   &  53   & 1.78   & 0.63    & 0.13     &0.08    & 91\\ 
			PFE  &284  & 304   &  54   &0.58   & 0.44    & 0.05    &0.05  & 90\\ 
			PG  &486  & 491   &  51   & 0.31   & 0.16    & 0.15     &0.11   & 44\\ 
			RTX  &822  & 831   &  51   & 1.74   & 1.5    & 0.93     &0.93  & 43\\ 
			TRV  &204  & 216   &  53   & 1.18   & 1.36    & 0.46     &0.49  & 63\\ 
			UNH  &750  & 618   &  41   & 2.05   & 1.23    & 0.54     &0.37  &72\\ 
			V  &416  & 413   &  50   & 0.51   & 0.43    & 0.22     &0.2  & 56\\ 
			VZ  &248  & 295   &  59   & 0.11   & 0.09    & 0.06    &0.06  & 43\\ 
			WBA  &744  & 641   &  43   &1.2   & 1.46    & 0.37     &0.38  & 72\\ 
			WMT  &494  & 475   &  48   & 0.47   & 0.25    & 0.14    &0.11  & 64\\ 
			XOM  &344  & 288   &  42   & 0.4   & 0.47    & 0.1     &0.11  &75\\ 
			\hline
			Total &16088  & 15277   &  49   & 0.93   & 0.68    & 0.23     &0.2   &73\\ 
			\hline\hline
		\end{tabular}
	\end{center}
\end{table}

\newpage
\subsection{Statistics of options on each stock using M.1 with $L=0$}
\begin{table}[H]
	\caption{Summary statistics (matching with $L=0$) of options on each stock in DJI by consolidating options related to the same underlying asset and expiration date.}
	\setlength\tabcolsep{3pt}
	\setlength{\textfloatsep}{3pt}
	\small
	\begin{center}
		\begin{tabular}{p{0.10\textwidth} | p{0.12\textwidth} p{0.15\textwidth} p{0.12\textwidth} | p{0.12\textwidth} p{0.12\textwidth}}
			\hline
			\rule{0pt}{10pt}
			Stock &\#markets &\#expirations &\#markets per expiration & \#matches ($L=0$) & average profit\\
			\hline\hline
			AAPL  & 1452  & 14   &  104   & 3   & 2.14    \\
			AXP  & 804  & 11   &  73   & 4   & 1.56    \\
			BA  & 1694  & 14   &  121   & 4   & 3.72    \\
			CAT  & 968  & 13   &  74   & 5   & 0.48    \\
			CSCO  &728  & 12   &  61   & 0   & 0    \\
			CVX   &742  & 12   &  62   & 5   & 0.36    \\
			DD   &778  & 13   &  60   & 2   & 0.03    \\
			DIS  &806  & 13   &  62   & 5   & 0.58    \\
			GS   &1110  & 12   &  93   & 3   & 7.3    \\
			HD   &908  & 13   &  70   & 1   & 0.29    \\
			IBM  &848  & 12   &  71   & 2   & 4.89    \\
			INTC  &662  & 12   &  55   & 1   & 0.01   \\
			JNJ  &850  & 12   &  71   & 2   & 0.08   \\
			JPM  &854  & 13   &  66   & 2   & 0.28   \\
			KO  &618  & 13  &  48   & 1   & 0.13   \\
			MCD  &692  & 12   & 58   & 2   & 0.43    \\
			MMM  &804  & 12   &  67   &1   & 1.36    \\
			MRK  &766  & 12   &  64   & 2   & 0.04    \\
			MSFT  &1194  & 13   & 92   & 0   & 0    \\
			NKE  &844  & 12   &  70   & 1   & 0.01    \\
			PFE  &640  & 12   &  53   & 3   & 0.15   \\
			PG  &786  & 12   &  66   & 3   & 0.48    \\
			RTX  &920  & 14   &  66   & 1   & 0.01   \\
			TRV  &256  & 6   &  43   & 0   & 0    \\
			UNH  &964  & 12   &  80   & 2   & 7.54    \\
			V  &856  & 13   &  66   & 5   & 3.5    \\
			VZ  &508  & 12   &  42   & 4   & 0.49   \\
			WBA  &808  & 11   &  73   & 0   & 0  \\
			WMT  &810  & 12   &  68   & 4   & 0.18 \\
			XOM  &832  & 12   &  69   & 6   & 1.16\\
			\hline
			Total &25502  & 366   &  69   & 74   & 1.54\\ 
			\hline\hline
		\end{tabular}
	\end{center}
\end{table}
\newpage
\begin{table}[H]
	\caption{Summary statistics (quoting with $L=0$) of options on each stock in DJI by consolidating options related to the same underlying asset and expiration date.}
	\small
	\setlength\tabcolsep{3pt}
	\setlength{\textfloatsep}{3pt}
	\begin{center}
		\begin{tabular}{p{0.08\textwidth} | p{0.08\textwidth} p{0.08\textwidth} p{0.08\textwidth} | p{0.08\textwidth} p{0.08\textwidth} p{0.12\textwidth} p{0.12\textwidth} p{0.10\textwidth}}
			\hline
			\rule{0pt}{10pt}
			Stock &\#option series &\#orders in $\F$ & Frontier (\%) & avg call spread & avg put spread & improved call spread & improved put spread & \% spread reduced\\
			\hline\hline
			AAPL  & 1120  & 902   &  40   & 1.61   & 2.05    & 0.69     & 1.1   & 51\\ 
			AXP  & 616  & 605   &  49   & 0.31   & 0.26    & 0.17     & 0.15   & 44\\ 
			BA  & 1208  & 1399   & 58   & 1.03   & 0.49   & 0.65     &0.35  & 34\\ 
			CAT  & 678  & 768   &  57   & 0.6   & 0.42    & 0.25     &0.19   & 57\\ 
			CSCO  &728  & 783   & 54   & 1.34   & 0.48    & 0.28     &0.14   & 77\\ 
			CVX   &456  & 537   &  59   & 0.48   & 0.58    & 0.23     &0.24    & 56\\ 
			DD   &666  & 845   &  63   & 0.28   & 0.2    & 0.16     &0.12  & 42\\ 
			DIS  &518  & 569   &  55   & 0.91   & 0.79    & 0.31    &0.24  & 68\\ 
			GS   &826  & 848   &  51   & 0.92   & 0.65   & 0.54     &0.3  & 46\\ 
			HD   &828  & 887   &  54   & 1.68   & 1.51    & 0.93    &0.83   & 45\\ 
			IBM  &696  &668   &  48   & 1.62   & 0.81    & 0.73     &0.34    & 56\\ 
			INTC  &568  & 701   &  62   & 1.02   & 0.61    & 0.32     &0.21   & 67\\ 
			JNJ  &672  & 779   &  58   & 0.99   & 1.16  & 0.39     &0.38  & 64\\ 
			JPM  &724  & 835   &  58   & 0.73   & 0.75  & 0.26     &0.25  & 66\\ 
			KO  &572  & 676  &  59   & 0.44   & 0.18    & 0.1     &0.07   & 73\\ 
			MCD  &558  & 710   & 64   & 0.63   & 0.34    & 0.35     &0.22    & 41\\ 
			MMM  &722  & 934   &  65   & 1.03   & 0.84    & 0.52     &0.44   & 49\\ 
			MRK  &672  & 791   &  59   & 0.36   & 0.28    & 0.17    &0.14   & 52\\ 
			MSFT  &1194  & 1003   & 42   & 2.19   & 0.97    & 1    &0.47   &  53\\ 
			NKE  &788  & 874   &  55   & 1.78   & 0.63    & 0.34     &0.14    & 80\\ 
			PFE  &480  & 612   &  64   &0.51   & 0.32    & 0.11    &0.09  & 76\\ 
			PG  &582  & 665  &  57   & 0.34   & 0.16    & 0.18     &0.12   & 40\\ 
			RTX  &822  & 931   &  57   & 1.74   & 1.5    & 1.33     &1.24  & 21\\ 
			TRV  &256  & 307   & 60   & 1.22   & 1.47    & 0.78     &0.82  & 41\\ 
			UNH  &806  & 722   &  45  & 2.02   & 1.22    & 1.15     &0.61  &46\\ 
			V  &580  & 689   &  59   & 0.45   & 0.38    & 0.27     &0.19  & 45\\ 
			VZ  &384  & 511   &  67   & 0.11   & 0.09    & 0.07    &0.06  & 35\\ 
			WBA  &808  & 783   &  48   &1.24   & 1.47    & 0.65     &0.77  &48\\ 
			WMT  &550  & 605   &  55   & 0.45   & 0.28    & 0.2    &0.14  & 53\\ 
			XOM  &440  & 459   &  52   & 0.4   & 0.5    & 0.5     &0.17  &64\\ 
			\hline
			Total &20518  & 22398   &  56   & 0.95   & 0.71    & 0.44     &0.35   &52 \\ 
			\hline\hline
		\end{tabular}
	\end{center}
\end{table}

\newpage
\section{Ethical Impact Statement}

The proposed financial market design has the potential to improve social welfare among traders and exchange owners by allowing traders to more precisely hedge their risks (e.g., in one transaction, purchase a put option on their exact investment portfolio), speculate on correlated movements among stocks, and produce an arbitrage surplus to split between the exchange and traders. Traders who currently specialize in identifying and exploiting arbitrage among related but independent markets may see their opportunities evaporate. To the extent that those specialists provide liquidity, exchanges may lose volume. With consolidated markets, profits go to informed traders. A consolidated or combinatorial market levels the playing field among traders so that information is primary, not computational power.

Trust is key in any market. An advantage of independent markets is their simplicity. A combinatorial market is in principal straightforward---it solves an optimization problem that can be openly published as a clear box---however, in practice, an average trader may have difficulty understanding what the optimization is doing. Complexity can lead to mistrust and conspiracy theories. The exchange must work hard to provide evidence that it is operating aboveboard, solving the exact optimization as promised. Publishing its algorithm and the de-identified orders is a good step toward accountability, but is not foolproof: the exchange could place shill orders to extract more surplus for itself. The exchange must educate traders and earn trust over time. The gain in social welfare should be enough to outweigh the downside of reduced trust.

In another sense, combinatorial markets are easier and more fair for naive traders than independent markets. Traders can order exactly what they want and know that they receive the best price even by consolidating orders from many markets in complex ways. Combinatorial markets neutralize the advantage of sophisticated traders who know how to manually consolidate orders on their own. In a zero-sum game like an options market, one trader's advantage is another trader's cost.

Allowing vastly more ways to bet on stock movements may hurt people who trade for 
the excitement of gambling, especially those prone to addiction. Allowing bets on combinations can introduce opportunities for fraud. For example, a fraudster can bet on an unusual point spread between sports teams and, behind the scenes, bribe players to take actions that are hard to notice and do not materially impact their team's chances of winning games and tournaments. Transparency will be key to mitigate such fraud. When all trades can be inspected by authorities and even the public, even in pseudo-anonymous form, suspicious activity is harder to hide and more likely to be uncovered.

\end{document}